\documentclass[a4paper,notitlepage,
floatfix,
aps,
pra,
superscriptaddress,
nofootinbib,
11pt]{revtex4-2}

\usepackage{graphicx}

\usepackage{setspace}
\usepackage{relsize}

\usepackage{amsmath}
\usepackage{amsfonts}
\usepackage{amssymb}
\usepackage{amscd}
\usepackage{mathtools}
\usepackage{amsthm}
\usepackage{mathrsfs}
\usepackage{dsfont}
\usepackage{braket}
\usepackage{bm}

\usepackage{enumerate}
\usepackage{paralist} 
\usepackage{framed}
\usepackage{verbatim}
\usepackage{url}
\usepackage{here}
\usepackage{hyperref}
\usepackage[dvipsnames]{xcolor}
\hypersetup{
    pdfnewwindow=true,
    colorlinks=true,
    citecolor=Blue,
    linkcolor=Blue,
    urlcolor=Blue
}

\usepackage{tikz}

\newtheorem{theorem}{Theorem}
\newcommand{\thmref}[1]{Theorem~\ref{#1}}
\newtheorem{proposition}{Proposition}
\newcommand{\propref}[1]{Proposition~\ref{#1}}

\newtheorem{lemma}{Lemma}

\theoremstyle{definition}
\newtheorem*{zremark}{Remark}
\newenvironment{tech-remark}{\par\small\zremark}{\endzremark}
\newtheorem{definition}{Definition}

\newcommand{\<}{\langle}
\renewcommand{\>}{\rangle}
\newcommand{\idch}{\mathrm{id}}

\DeclareMathOperator\supp{supp}
\DeclareMathOperator{\Tr}{Tr}
\newcommand{\QC}{\operatorname{QC}}
\newcommand{\GO}{\operatorname{GO}}

\newcommand{\ext}{\mathrm{ext}}

\newcommand{\meas}{\mathrm{meas}}
\newcommand{\eras}{\mathrm{eras}}

\newcommand{\In}{\mathrm{in}}
\newcommand{\Eq}{\mathrm{eq}}

\newcommand{\cA}{\mathcal{A}}

\newcommand{\cF}{\mathcal{F}}

\newcommand{\cH}{\mathcal{H}}

\newcommand{\cK}{\mathcal{K}}
\newcommand{\cL}{\mathcal{L}}
\newcommand{\cM}{\mathcal{M}}

\newcommand{\cU}{\mathcal{U}}
\newcommand{\cV}{\mathcal{V}}

\newcommand{\M}{\mathsf{M}}
\newcommand{\A}{\mathsf{A}}

\newcommand{\one}{\mathds{1}}
\newcommand{\zero}{\mathds{O}}

\renewcommand{\ge}{\geqslant}
\renewcommand{\le}{\leqslant}

\newcommand{\eq}[1]{Eq.~\eqref{#1}}
\newcommand{\rank}[1]{\mathrm{rank}\left( {#1}\right)}

\begin{document}
\title{Universal validity of the second law of information thermodynamics}
\author{Shintaro Minagawa}
\email{minagawa.shintaro@nagoya-u.jp}
\affiliation{Graduate School of Informatics, Nagoya University, Furo-cho, Chikusa-Ku, Nagoya 464-8601, Japan}
\author{M. Hamed Mohammady  }
\email{m.hamed.mohammady@savba.sk}
\affiliation{QuIC, \'{E}cole Polytechnique de Bruxelles, CP 165/59, Universit\'{e} Libre de Bruxelles, 1050 Brussels, Belgium}
\affiliation{RCQI, Institute of Physics, Slovak Academy of Sciences, D\'ubravsk\'a cesta 9, Bratislava 84511, Slovakia}
\author{Kenta Sakai}
\email{sakai.kenta\_32@nagoya-u.jp}
\affiliation{Graduate School of Informatics, Nagoya University, Furo-cho, Chikusa-Ku, Nagoya 464-8601, Japan}
\affiliation{(until March 2023)}
\author{Kohtaro Kato}
\email{kokato@i.nagoya-u.ac.jp}
\affiliation{Graduate School of Informatics, Nagoya University, Furo-cho, Chikusa-Ku, Nagoya 464-8601, Japan}
\author{Francesco Buscemi}
\email{buscemi@nagoya-u.jp}
\affiliation{Graduate School of Informatics, Nagoya University, Furo-cho, Chikusa-Ku, Nagoya 464-8601, Japan}
\begin{abstract}
    Adiabatic measurements, followed by feedback and erasure protocols, have often been considered as a model to embody Maxwell's Demon paradox and to study the interplay between thermodynamics and information processing. Such studies have led to the conclusion, now widely accepted in the community, that Maxwell's Demon and the second law of thermodynamics can peacefully coexist because any gain provided by the demon must be offset by the cost of performing the measurement and resetting the demon's memory to its initial state. Statements of this kind are collectively referred to as \emph{second laws of information thermodynamics} and have recently been extended to include quantum theoretical scenarios. However, previous studies in this direction have made several assumptions, particularly about the feedback process and the demon's memory readout, and thus arrived at statements that are not universally applicable and whose range of validity is not clear. In this work, we fill this gap by precisely characterizing the full range of quantum feedback control and erasure protocols that are overall consistent with the second law of thermodynamics.  This leads us to conclude that the second law of information thermodynamics is indeed \textit{universal}: it must hold for any quantum feedback control and erasure protocol, regardless of the measurement process involved, as long as the protocol is overall compatible with thermodynamics. Our comprehensive analysis not only encompasses new scenarios but also retrieves previous ones, doing so with fewer assumptions. This simplification contributes to a clearer understanding of the theory. 
\end{abstract}

\maketitle

\section{Introduction}
The problem of consistency between the second law of thermodynamics and information processing has been at the center of one of the longest running debates in the history of modern physics, ever since Maxwell conjured up his famous demon~\cite{maxwell1871theory}.
A widely accepted solution to Maxwell's paradox is that consistency with the second law of thermodynamics is recovered by taking into account the work cost for measurement and erasure, i.e., the resetting of the demon's memory to its initial state~\cite{szilard1929uber,brillouin1951maxwell,Landauer,bennett1973logical,BennettComp,LeffRex,maruyama2009colloquium}.
These ideas, bridging thermodynamics with information theory, are nowadays collectively referred to as \emph{information thermodynamics}~\cite{ingarden-1997,sagawa2013information}.

In this context, and including a quantum theoretical scenario, Sagawa and Ueda, in a series of celebrated papers~\cite{sagawa-ueda2008second,sagawa2009minimal,sagawa2009erratum}, derived an achievable upper bound for the work extracted by feedback control and showed that the conventional second law can, in general, be violated from the viewpoint of the system alone, but such a violation is exactly compensated by the cost of implementing the controlling measurement and resetting the memory.
Such a tradeoff relation is what they call \emph{the second law of information thermodynamics (ITh)}.

Unfortunately, despite their importance, the balance equations established in Refs.~\cite{sagawa-ueda2008second,sagawa2009minimal,sagawa2009erratum} rely on several mutually inconsistent assumptions that lack a direct operational interpretation. Moreover, these works only discuss \emph{sufficient} conditions for the validity of such balance equations. While some generalizations and refinements have been proposed~\cite{Jacobs2009,funo2013integral,abdelkhalek2016fundamental, strasberg17,Mohammady2019c, Strasberg2019, Strasberg2020b, strasberg2022quantum, Latune2024}, the demon's memory readout process is always limited to \textit{ideal projective measurements}. Besides being unrealistic in practice, such an assumption is problematic \textit{in principle}: since the demon's memory enters directly into the thermodynamic balance, the process acting on it must be treated \textit{in full generality}, lest we obtain statements of limited scope. As a result, a comprehensive characterization of the validity range of the second law of ITh remains elusive, and it is unclear under what conditions the second law of ITh holds. In fact, at the time of writing, it is not even clear whether the second law of ITh should be considered a universal law or not, and what its logical status is with respect to the conventional second law of thermodynamics.

Our paper addresses this gap by adopting a top-down approach. Instead of attempting to \emph{derive} the second law from assumptions with unclear logical necessity, we initiate from a purely information-theoretic framework and obtain balance equations that hold \textit{for any measurement and isothermal feedback process}, in particular including any readout mechanism, and subsequently \emph{impose} the second law of phenomenological thermodynamics as a constraint. This approach, which follows that used by von Neumann to derive his entropy's equation~\cite{von1955mathematical,minagawa-2022}, allows us to determine exactly (in terms of sufficient \textit{and necessary} conditions) how far feedback control and erasure protocols can be generalized while remaining overall consistent with the second law.  We are then able to demonstrate the universal validity of the second law of ITh in general feedback control and erasure protocols: as long as such a protocol is compatible with the  second law of phenomenological thermodynamics, it must also satisfy the second law of ITh, regardless of the measurement and feedback process involved.

A quantity that plays a crucial role in our analysis is the Groenewold--Ozawa information gain \cite{groenewold1971problem,ozawa-1986-groen-info}: while previous works \cite{Jacobs2009, funo2013integral, abdelkhalek2016fundamental, Strasberg2019,danageozian2021thermodynamic,Mohammady2022} have also provided it with a thermodynamic interpretation---even in situations when it takes negative values---our balance equations show that such interpretation holds in complete generality.

\section{Results}

\begin{figure}[t]
    \centering
    \includegraphics[width=10cm]{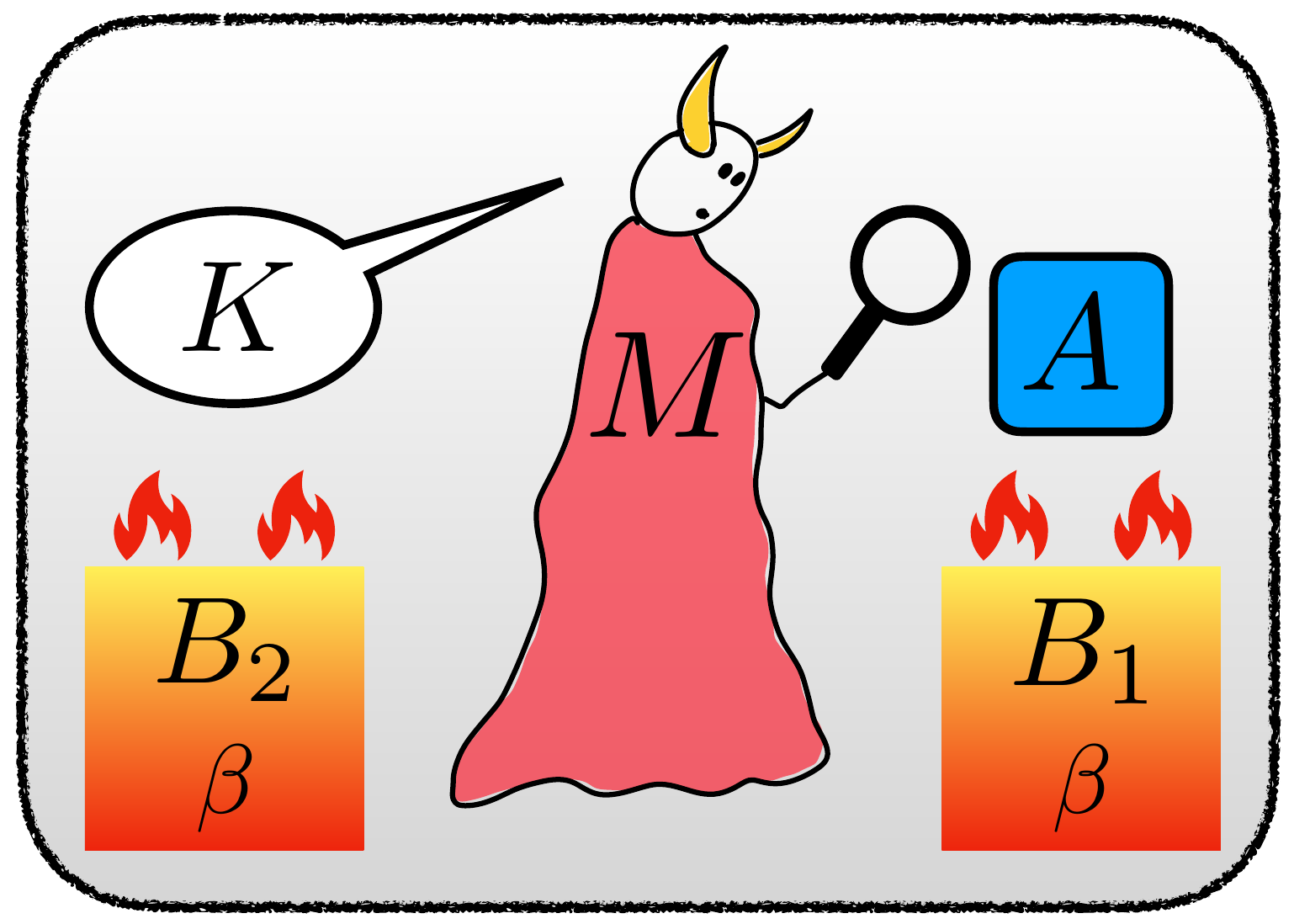}
    \caption{The systems appearing in our setup: the target system $A$, the controller (demon) consisting of an internal state $M$ and a classical register $K$, and two baths $B_1$ and $B_2$ at the same inverse temperature $\beta$.}
    \label{fig:schematic}
\end{figure}

The minimum scenario required to discuss Maxwell's paradox and feedback control protocols in full generality, but without oversimplifications, comprises five systems, as shown in Fig.~\ref{fig:schematic}: the physical system (i.e., the gas) being measured, denoted by $A$; the controller's (i.e., the demon's) internal state $M$ (where the letter ``M'' stands for ``Maxwell'', ``measurement apparatus'' or ``memory''); a classical register $K$ recording the measurement's outcomes; and two independent baths $B_1$ and $B_2$ (one used during the feedback control stage, the other used for the final erasure of the measurement), which are assumed to be at the same finite temperature. This means that the overall process is assumed to be isothermal.

Without any feedback control, for isothermal processes, the second law of thermodynamics is equivalent to the statement that the work extracted from the system $A$ can reach but not exceed the change in the free energy\footnote{These and other key concepts will be rigorously introduced and discussed in what follows. The purpose of these first few paragraphs is simply to provide a relatively informal overview of our main findings.} of the system---in formula, $W^A_\ext\le-\Delta F^A$. The main contribution of Ref.~\cite{sagawa-ueda2008second} was to show that if feedback control is allowed instead, the work extracted can go all the way up to $W^A_\ext=-\Delta F^A+\beta^{-1}I_{\QC}$, where $I_{\QC}$ is a non-negative term quantifying the amount of information collected by the measurement used to guide the subsequent feedback control protocol. In this sense, Maxwell's demon \textit{can} indeed violate the second law of thermodynamics, but this conclusion should come as no surprise, since the demon is not yet included in the global thermodynamic balance at this point.

Indeed, once the demon itself is embodied in a physical system, such a violation of the second law turns out to be only a \textit{local} violation, which is perfectly possible as long as it is compensated for elsewhere. According to Landauer's principle, such a compensation should be identified with the cost of performing the measurement and resetting the measurement apparatus and register at the end of the protocol, so that they are ready for use in the next round. Following this narrative, Refs.~\cite{sagawa2009minimal,sagawa2009erratum} list a number of assumptions about the quantum feedback protocol so that, as one would expect, the work cost of implementing the measurement and performing its erasure is lower bounded as $W^{MK}_\In\ge\beta^{-1}I_{\QC}$, thus guaranteeing that the total net work extracted $W_{\textrm{tot}}: =W^A_\ext-W^{MK}_\In\le -\Delta F^A$ is still within the limits of the second law of thermodynamics.

Our analysis begins by removing all assumptions from the consistency argument above. We argue that this is not just for the sake of mathematical generality, but is \textit{necessary} for two reasons. The first reason is that, specifically in relation to Refs.~\cite{sagawa-ueda2008second,sagawa2009minimal,sagawa2009erratum}, some of the assumptions made therein are, as we will show in what follows, extremely restrictive---so stringent, in fact, that they are inconsistent in most cases, constraining the analysis to trivial situations. The second reason is a matter of principle: if certain assumptions are required to restore the validity of the second law, the consistency between thermodynamics and quantum information processing cannot be considered universal, contrary to what folklore claims.

We then show that, when all assumptions about the mathematical form of the quantum feedback protocol are removed, the work extracted from the target system is upper bounded as
\begin{align}\label{eq:first-result}
W^A_\ext\le-\Delta F^A_{0\to4}+\beta^{-1}I_{\GO}\;,
\end{align}
while the work cost of implementing the measurement and its erasure is now lower bounded as
\begin{align}\label{eq:second-result}
W^{MK}_\In\ge\beta^{-1}[\Delta S^{AMK} + I_{\GO}]\;,
\end{align}
where $\Delta S^{AMK}$ denotes the  entropy change of the entire compound $AMK$ due to the measurement process and $I_{\GO}$ is the \textit{Groenewold--Ozawa information gain}~\cite{groenewold1971problem,ozawa-1986-groen-info}.
Note that while the bound~\eqref{eq:first-result} looks similar to the one given in~\cite{sagawa-ueda2008second}, the information quantity $I_{\GO}$ appearing in our bounds is different from the one used in Refs.~\cite{sagawa-ueda2008second,sagawa2009minimal,sagawa2009erratum}: in general, $I_{\GO}\gtreqqless I_{\QC}$. But while $I_{\QC}$ does not provide the correct bounds in general, $I_{\GO}$ does and, moreover, gives the same numerical values as $I_{\QC}$ in all cases considered in~\cite{sagawa-ueda2008second,sagawa2009minimal,sagawa2009erratum}.
Further, Eqs.~\eqref{eq:first-result} and~\eqref{eq:second-result} together imply that the net work extracted in general is bounded as
\begin{align}\label{eq:intermediate-overview}
    W_{\textrm{tot}}: =W^A_\ext-W^{MK}_\In\le -\Delta F^A-\beta^{-1}\Delta S^{AMK}\;.
\end{align}
In other words, even if the final erasure is implemented in accordance with Landauer's principle, the second law may still be violated whenever $\Delta S^{AMK}<0$.
Eqs.~\eqref{eq:first-result} and~\eqref{eq:second-result} constitute the main technical contributions of this work: their formal statement is given as Theorem~\ref{theorem:workbound} below.

Finally, by means of explicit counterexamples, we show that the axioms of quantum theory \emph{by themselves} are perfectly consistent with a measurement process that decreases the total entropy of the system-memory-register compound, implying a violation of the second law according to Eq.~\eqref{eq:intermediate-overview}. This leads us to the main conceptual contribution of this work, i.e. the conclusion that---contrary to some cursory accounts---in a quantum mechanical feedback process it is not enough to eventually perform an erasure process, as stipulated by Landauer's principle, to guarantee the validity of the second law. In other words, \textit{the second law of thermodynamics is logically independent of the axioms of quantum theory}, and its role is to constrain the set of possible measurement processes from the outset. Any attempt to \emph{prove} the second law from \emph{within} quantum theory is doomed to result in pure tautology~\cite{EARMAN1998435,Earman1999-EAREXT-4}.

\subsection{Framework}

Consider a quantum system $Y$ associated with a finite-dimensional Hilbert space $\mathcal H^Y$. The algebra of linear operators $L^Y$ on $\cH^Y$ will be denoted as $\cL(\cH^Y)$, $\one^Y$ and $\zero^Y$ denoting the unit and null  operators, respectively. States on $Y$ are represented by unit-trace positive operators, i.e., $\rho^Y \ge \zero^Y$,  $\Tr[\rho^Y]=1$.  A \emph{thermodynamic system} $Y$ is defined as the tuple $(\rho^Y; H^Y; \beta)$, where  $H^Y$ is the Hamiltonian and  $\beta:=1/k_BT>0$ is the inverse temperature of an external thermal bath,  with $k_B$  Boltzmann's constant. Throughout, we shall only consider the case where the thermal bath has a constant temperature, and so for notational simplicity we will abbreviate the thermodynamic system as $(\rho^Y; H^Y)$.  When the system is in thermal equilibrium, the  \emph{thermal state} or \emph{Gibbs state} is defined as $\gamma^Y:=e^{-\beta H^Y}/Z^Y $, where $Z^Y :=\Tr[e^{-\beta H^Y}]$ is the partition function.

\begin{figure}[t]
    \centering
    \includegraphics[width=12cm]{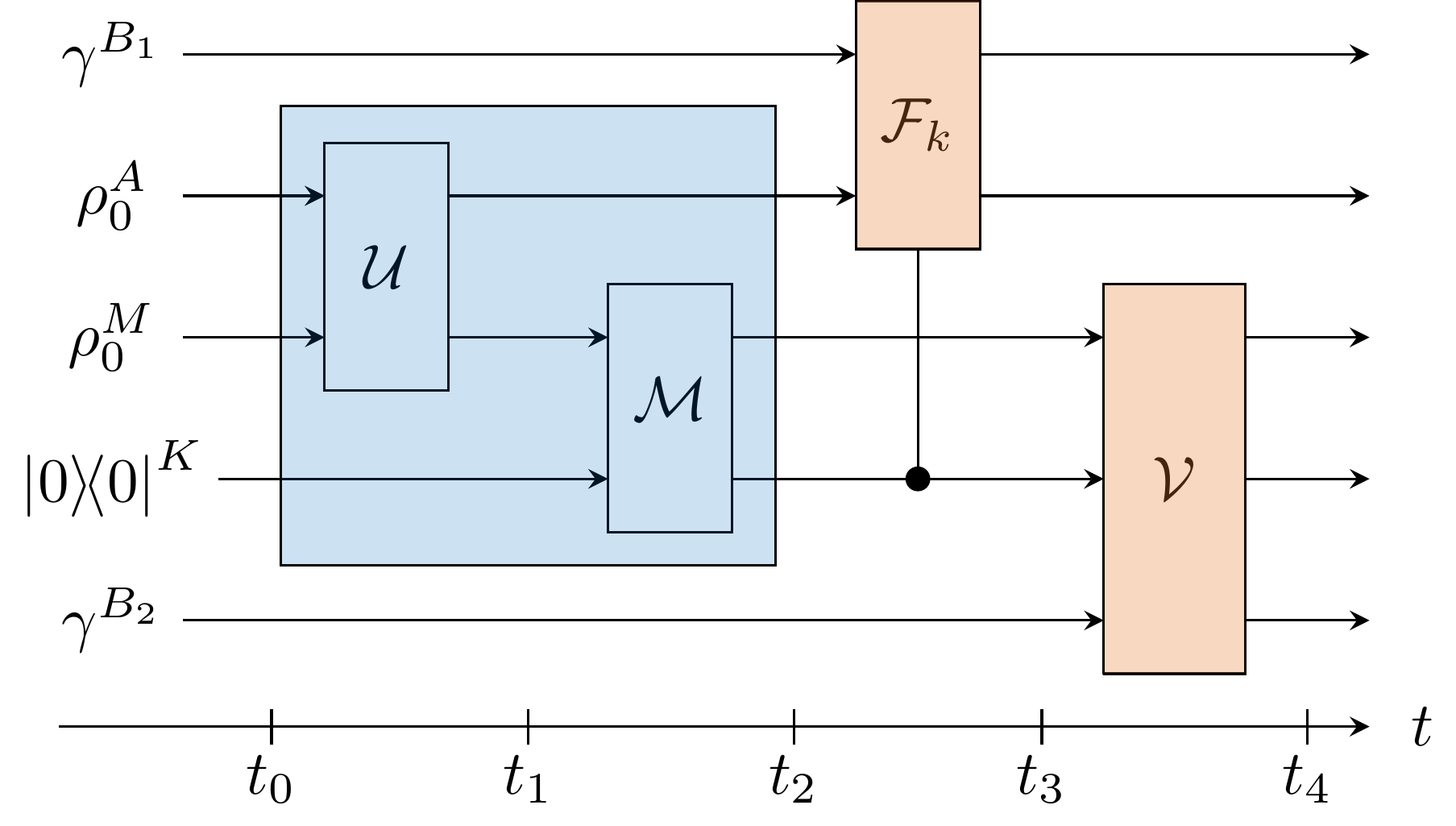}
    \caption{The circuit representation of a general quantum feedback control and erasure protocol. 
    \textbf{Interaction step} ($t_0\to t_1$): system $A$ and memory $M$ interact by a unitary channel $\cU$. 
    \textbf{Readout step} ($t_1\to t_2$): an instrument $\cM$ is applied on the memory $M$  and the outcome $k$ is written on the classical register $K$. 
    The interaction step and the readout step together are referred to as the measurement step.
    \textbf{Feedback control step} ($t_2\to t_3$): a controlled unitary channel $\cF_k$ is applied on the compound of system $A$ and thermal bath $B_1$ depending on the outcome $k$.
    \textbf{Erasure step} ($t_3\to t_4$): a unitary channel $\cV$ is applied on the compound of $MK$ and thermal bath $B_2$, so as to return the state of $MK$ to its initial configuration. The total compound system is assumed to evolve adiabatically during the entire protocol, that is, no heat is exchanged with any outside source.}
    \label{fig:process}
\end{figure}

The generalized quantum feedback control and erasure protocols we shall consider will  comprise of five discrete time steps $t_i$, $i=0,1,2,3,4$. The total system is composed of a target system $A$, a controller  consisting of a memory $M$ and a classical register $K$, and two thermal baths $B_1, B_2$, both of which have the same inverse temperature $\beta>0$, as depicted in Figure~\ref{fig:schematic}.  For notational simplicity, we shall omit superscripts when denoting any quantity pertaining to the \emph{entire} compound $B_1 A M K B_2$, reserving their use only when discussing subsystems; for example, the state of subsystem $AMK$ at time step $t_i$ will be denoted as $\rho^{AMK}_i := \Tr_{B_1 B_2}[\rho_i]$, etc.  In particular, we shall assume that the Hamiltonian at time step $t_i$ reads $H_i = H^{B_1} + H_i^A + H^{MK} + H^{B_2}$. That is, at each time step we assume that there are no interaction terms between the different subsystems, and only the Hamiltonian of the target system $A$ may change. The protocol is represented schematically in Figure~\ref{fig:process}; below we shall describe each step in detail.

\bigskip\indent {\bf The preparation step.}
At the initial time  $t=t_0$, the compound system is prepared in the state
\begin{equation}\label{eq:initial-state-total}
  \rho_0 := \gamma^{B_1} \otimes \rho^A_0\otimes\rho^M_0\otimes|0\>\<0|^K\otimes\gamma^{B_2}
\end{equation}
where $\rho^A_0$ and $\rho^M_0$ are arbitrary states on $A$ and $M$, respectively, while $\ket{0}^K$ represents the idle state of the classical register\footnote{
Note that the memory considered in Ref.~\cite{sagawa2009minimal} is described by a Hilbert space with a direct sum structure.
Here we describe the degrees of freedom of the labels of the blocks and the internal states of the memory using different quantum systems.
In the context of our paper, the two pictures are clearly equivalent.}, and $\gamma^{B_1}, \gamma^{B_2}$ are the thermal states of the baths, with respect to the same inverse temperature $\beta$. Note that a common assumption is that the initial state of the memory $\rho_0^M$ is thermal at the same inverse temperature $\beta$ as the two baths: while such an assumption is very reasonable from a physical point of view, and in particular facilitates the discussion of the erasure step (see below), for the sake of generality we keep $\rho_0^M$ arbitrary.

\bigskip \indent {\bf The measurement step.}
This step comprises an interaction step and a readout step. The \emph{interaction} or \emph{pre-measurement} step (from $t=t_0$ to $t=t_1$) represents the interaction between $A$ and $M$, described by a unitary channel $\cU(\cdot):=U(\cdot) U^\dagger$ acting in $AM$. The \emph{readout} or \emph{pointer objectification} step (from $t=t_1$ to $t=t_2$) is represented as a \emph{CP-instrument}~\cite{ozawa1984quantum} acting in $M$, namely, a family $\cM := \{\cM_k: k\in\cK \}$ of completely positive linear maps $\cM_k: \cL(\cH^M) \to \cL(\cH^M)$, labeled by the measurement outcomes $k \in \cK$,  such that their sum $\cM_{\cK} := \sum_{k\in\cK}\cM_k$ is trace-preserving, i.e., a channel. The instrument $\cM$ is associated with a unique positive operator-valued measure (POVM) $\M := \{\M_k : k \in \cK\}$, with elements defined using the ``Heisenberg picture'' dual of $\cM_k$ as $\M_k := \cM_k^*(\one^M)$. Since the POVM $\M$ acts in the memory, it is referred to as the \emph{pointer observable}.  After $M$ is measured by the instrument $\cM$, the observed outcome $k$ is recorded in the classical register. Such classical readouts are assumed to be all perfectly distinguishable, and thus are represented, following a common convention in quantum information theory~\cite{wilde2017quantum-book}, by orthogonal pure states $\ket{k}^K$.

Accordingly, at $t=t_2$ the state of the compound system reads 
\begin{align}\label{eq:state-after-measurement}
 \rho_2 &:= \gamma^{B_1} \otimes \bigg( \sum_{k\in\cK}(\mathrm{id}^A\otimes\cM_k) \Big[\cU (\rho^{A}_0 \otimes \rho^M_0)\Big]\otimes\ket{k}\!\!\bra{k}^K\bigg) \otimes \gamma^{B_2} =:\sum_{k \in \cK} p_k \, \rho_{2,k}\;,
\end{align}
where $\idch^A$ denotes the identity channel acting in $A$, and 
\begin{align*}
    \rho_{2,k} :=  \gamma^{B_1} \otimes \rho^{AM}_{2,k} \otimes |k\>\<k|^K \otimes \gamma^{B_2}\;,
\end{align*}
with
\begin{equation*}
    \rho^{AM}_{2,k} :=\frac{(\mathrm{id}^A\otimes\cM_k) \Big[\cU (\rho^{A}_0 \otimes \rho^M_0)\Big]}{p_k}\;
\end{equation*}
whenever the probability of obtaining  outcome $k$ satisfies
\begin{equation*}
    p_k:=\Tr\bigg\{(\mathrm{id}^A\otimes\cM_k) \Big[\cU (\rho^{A}_0 \otimes \rho^M_0)\Big]\bigg\}>0\;,
\end{equation*}
otherwise $\rho^{AM}_{2,k}$ can be defined arbitrarily.

We note that a fixed tuple $(\cH^M, \rho^M_0, \cU,\cM)$ defines a \emph{measurement process} or \emph{measurement scheme} for an instrument $\cA := \{\cA_k : k \in \cK\}$ acting in the target system $A$, with the operations reading
\begin{align}\label{eq:system-instrument}
    \cA_k(\cdot) := \Tr_{M}\bigg\{(\idch^A \otimes \cM_k)  \Big[\cU(\cdot \otimes \rho^M_0)\Big]\bigg\} \equiv \Tr_{M}\Big[\one^A \otimes \M_k\ \ \cU(\cdot \otimes \rho^M_0)\Big].
\end{align}
In particular, we stress that an instrument on the target system $\cA := \{\cA_k : k \in \cK\}$ can be realized by means of infinitely many different measurement processes. One of the results of this work will be to show that the laws of thermodynamics constrain the latter, not the former.

\begin{tech-remark}
The formalism of CP-instruments provides the most general readout (i.e., pointer objectification) procedure allowed by quantum theory. While general instruments in the target system $A$ have been considered before, all previous works have focused on a restricted class of instruments acting in the memory $M$, namely, L\"uders instruments compatible with a projection-valued measure (PVM), also known as ``ideal projective measurements'' ~\cite{sagawa2009minimal,sagawa2009erratum, Jacobs2009,funo2013integral,abdelkhalek2016fundamental,strasberg17,Mohammady2019c, Strasberg2019, Strasberg2020b,  strasberg2022quantum, Latune2024}. 
 $\M$ is a PVM if the effects  $\M_k$  are mutually orthogonal projections, and the operations of the corresponding $\M$-compatible L\"uders instrument read   $\cM_k^L(\cdot) := \M_k (\cdot) \M_k$.  As shown by Ozawa~\cite{ozawa1984quantum},  every instrument acting in $A$ admits a canonical measurement scheme, where $\rho^M_0$ is chosen to be pure and the pointer observable is chosen to be a PVM. But we stress that the pointer observable in a given measurement process need not be a PVM; and, even if it is, the instrument measuring it need not be of the L\"uders form. In fact,   it is well known that every observable  $\M$ admits infinitely many $\M$-compatible instruments.   
\end{tech-remark}

\bigskip\indent {\bf The feedback control step.}
From $t=t_2$ to $t=t_3$, a feedback control protocol is performed. This is implemented by coupling the compound $AK$ with the thermal bath $B_1$ by  a  unitary channel $\cF(\cdot) := F (\cdot) F^\dagger$, defined by the unitary operator \footnote{Note that unitarity of $F$ implicitly assumes that $K$ is represented by a Hilbert space $\cH^K$ of dimension equal to the number of measurement outcomes, i.e, $\dim(\cH^K) = |\cK|$.} 
\begin{align*}
    F := \sum_{k\in \cK} F_k \otimes  |k\>\<k|^K. 
\end{align*}
Here, $F_k$ are unitary operators  on  $B_1A$, which induce the unitary channel $\cF_k (\cdot) := F_k (\cdot) F_k^\dagger$  conditional on the classical register having recorded outcome $k$. At time step $t=t_3$, the state of the compound reads
\begin{equation}\label{eq:state-after-feedback}
  \rho_3 := (\cF^{B_1AK} \otimes \idch^{MB_2} )(\rho_2) =   \sum_{k\in \cK} p_k\, \rho_{3,k}\;,    
\end{equation}
where 
\begin{align*}
 \rho_{3,k} = \rho_{3,k}^{B_1AM}\otimes\ket{k}\!\!\bra{k}^K\otimes \gamma^{B_2}\;.
\end{align*}
Here, $\rho^{B_1AM}_{3,k}=(\cF_k^{B_1A}\otimes\mathrm{id}^M)(\gamma^{B_1} \otimes \rho^{AM}_{2,k})$. We shall say that the feedback process is \emph{pure unitary} if we choose $F_k = \one^{B_1} \otimes F_k^A$, so that for each outcome the target system undergoes an isolated unitary evolution. In other words, a pure unitary feedback process does not involve the thermal bath. This is the case considered in, e.g., Refs.~\cite{sagawa-ueda2008second,sagawa2009minimal,sagawa2009erratum}. However, since Szilard~\cite{szilard1929uber} onward, the traditional formulation typically considers a feedback protocol that is done in contact with a thermal bath, as we do here.

\bigskip\indent {\bf The erasure step.} 
Lastly, the erasure process from $t=t_3$ to $t=t_4$ is modeled by coupling $MK$ with the thermal bath $B_2$ by a unitary channel $\cV(\cdot):=V(\cdot)V^\dagger$. We naturally assume that $H^A_3=H^A_4$, since the target system $A$ remains dormant. At time step $t_4$, the state of the compound system will read
\begin{equation}\label{eq:erasure}
\rho_4 := (\idch^{B_1 A} \otimes \cV^{MKB_2})(\rho_3)\;,   
\end{equation}
such that, by definition of ``erasure'', $\rho^{MK}_4 = \rho^{MK}_0 = \rho^M_0 \otimes |0\>\<0|^K$. That is, the interaction between $MK$ and the bath $B_2$ returns the local state of $MK$ back to its initial configuration. Such a setting appears in the context of \emph{Landauer's principle}~\cite{Landauer,reeb2014improved,buscemi2020thermodynamic}. If, in addition, it holds that $\rho^{AMK}_4 = \rho_4^A \otimes \rho^M_0 \otimes |0\>\<0|^K $, i.e., if the correlations between $A$ and $MK$ are also erased,  then we say that the erasure is \emph{perfect}. Otherwise, we call the erasure \emph{partial}.   While in principle perfect erasure can always be achieved if a suitable bath is provided, it is a non-trivial problem to determine whether such a unitary erasure process always exists for a \textit{given} bath.  To alleviate this problem, we also consider here protocols that include \textit{partial} erasure. As mentioned above when discussing the preparation step, a conceptually simpler situation occurs when the initial state of the memory is thermal at the same bath temperature, so that the erasure process can be intuitively understood as a thermalization process.

\subsection{About injected and extracted work, and the assumption of overall adiabaticity}

The internal energy of a thermodynamic system  is $E(\rho^Y; H^Y) := \Tr[\rho^Y H^Y]$, and the non-equilibrium free energy~\cite{gaveau1997general,esposito2011second} is $F(\rho^Y; H^Y):= E(\rho^Y; H^Y)-\beta^{-1}S(Y)_\rho$, where $S(Y)_\rho :=-\Tr[\rho^Y\ln\rho^Y]$ is the von Neumann entropy~\cite{von1955mathematical}.  When a thermodynamic system  transforms from $t=t_i$ to $t=t_j$ as $(\rho^Y_i; H^Y_i) \mapsto (\rho^Y_j; H^Y_j)$,   we denote the increase in internal energy $E$,  nonequilibrium free energy $F$,  and entropy $S$ as follows:
\begin{equation}
    \Delta x^{Y}_{i\to j}:=x(\rho^Y_j;H^Y_j)-x(\rho^Y_i;H^Y_i)\quad(x=E, F ,S)\;.
\end{equation}

\begin{definition}\label{defn:work-adiabatic}
Consider a thermodynamic system which transforms as $(\rho^Y_i; H^Y_i) \mapsto (\rho^Y_j; H^Y_j)$. The transformation is defined as \emph{adiabatic} if it does not involve an exchange of heat with an external bath.   In such a case, by the first law of thermodynamics, the work injected into (resp., extracted from) the system is defined as the increase (resp., decrease) in internal energy, i.e., 
 \begin{align*}
   W_\In^Y \equiv -W_\ext^Y:= \Delta E^Y_{i \to j}\;.
\end{align*}
 \unskip\nobreak\hfill $\square$
\end{definition}

In our formalism, all thermal baths (i.e., the systems $B_1$ and $B_2$) are treated as \emph{internal} and so there are no \emph{external} baths with which heat is exchanged. Moreover, following a well-established convention dating back to Szilard~\cite{szilard1929uber} and von Neumann~\cite{von1955mathematical}, and routinely adopted until these days~\cite{sagawa-ueda2008second,sagawa2009minimal,abdelkhalek2016fundamental,Mancino2017,  Purves-2020, Panda2023}, we assume that the pointer objectification implemented by the instrument $\cM$ is also adiabatic, although it is obviously non-unitary. This may be justified if, for example, the objectification process is sufficiently fast  with respect to the time scale required for heat to dissipate~\cite{Smith_2018}.
Concerning the rest, i.e., during the premeasurement, feedback, and erasure steps of the protocol, the total compound transforms  by a global unitary channel which, by definition, does not involve an interaction with \textit{any} external system, and so clearly no heat is exchanged here either.
In conclusion, while the subsystem $AMK$ exchanges heat with $B_1$ and $B_2$ during the feedback and erasure steps, respectively, we treat the total compound $B_1 AMK B_2$ as transforming adiabatically during the entire protocol.

Since the total process is adiabatic,  the net extracted work  is identified with the decrease in internal energy of the entire compound, that is, $W_{\textrm{tot}} = -\Delta E_{0 \to 4}$.
Now we wish to split the contribution to the total work as that originating from the target system $A$ and that originating from the controller $MK$. To this end, we note that the target system is involved only during the measurement  and feedback steps,  the controller is involved only during the measurement  and erasure steps, the thermal bath $B_1$ is involved only during the feedback step, and the thermal bath $B_2$ is involved only during  the erasure step. As such, we may write (see Methods, Section \ref{app:preliminaries})
\begin{align}\label{eq:net-work-extracted}
    W_{\textrm{tot}} &= -\Delta E_{0 \to 4} \nonumber \\
    & = -\Delta E_{0 \to 2} -\Delta E_{2 \to 3} -\Delta E_{3 \to 4} \nonumber  \\
    & = -\Delta E_{0 \to 2}^A  -\Delta E_{0 \to 2}^{MK} -\Delta E_{2 \to 3}^{B_1 A}  -\Delta E_{3 \to 4}^{MK B_2} \nonumber \\
    & = W_\ext^{A} - W_\In^{MK}, 
\end{align}
where    
\begin{align}\label{eq:work-extracted-A}
W^{A}_\ext&:=  -\Delta E^A_{0 \to 2} - \Delta E_{2 \to 3} \equiv -\Delta E^A_{0 \to 2} - \Delta E_{2 \to 3}^{B_1 A}  \;
\end{align}
is the  work extracted from the target system, and 
\begin{equation}\label{eq:work-injected-MK}
W^{MK}_\In:= \Delta E^{MK}_{0 \to 2} + \Delta E_{3 \to 4} \equiv \Delta E^{MK}_{0 \to 2} + \Delta E_{3 \to 4}^{MKB_2} \;
\end{equation}
is the work injected into the controller.

\subsection{General work bounds}
Before providing general bounds for the work defined in Eqs.~\eqref{eq:work-extracted-A} and \eqref{eq:work-injected-MK}, let us first introduce  some useful information-theoretic  quantities.  For any state $\rho^A$ and a positive operator $\sigma^A$ such that $\supp(\rho^A) \subseteq \supp(\sigma ^A) $,   the \emph{Umegaki quantum relative entropy} is defined by $D(\rho^A\|\sigma ^A):=\Tr[\rho^A( \ln\rho^A-\ln\sigma ^A)] \geqslant 0$ \cite{umegaki-q-rel-ent-1961}, which is non-negative due to Klein's inequality~\cite{klein1931zur,nielsen_chuang_2010}, and vanishes if and only if $\rho^A = \sigma ^A$. The \textit{quantum mutual information} of a bipartite state $\rho^{AB}$ is defined as $I(A\!:\!B)_\rho:=S(A)_\rho+S(B)_\rho-S(AB)_\rho \equiv D(\rho^{AB} \| \rho^A \otimes \rho^B) \ge 0$, with equality if and only if $\rho^{AB}=\rho^A \otimes \rho^B$. On the other hand, the \textit{conditional quantum entropy}  of a bipartite state $\rho^{AB}$ is defined as $S(A|B)_{\rho}:=S(AB)_{\rho}-S(B)_{\rho}$, which can be negative.  The \textit{conditional quantum mutual information} of a tripartite state $\rho^{ABC}$ is defined as $I(A\!:\!C|B)_\rho:=S(A|B)_{\rho}+S(C|B)_{\rho}-S(AC|B)_{\rho} \geqslant 0$, where the non-negativity follows from the strong subadditivity of the von Neumann entropy (see, e.g., Ref.~\cite{wilde2017quantum-book}). Finally, we introduce the  following information measure related to the measurement process on the target system:
\begin{definition}
    The \emph{Groenewold--Ozawa information gain}~\cite{groenewold1971problem,ozawa-1986-groen-info} of the target system's measurement process is defined as:
    \begin{equation}
        I_{\GO}:=S(A)_{\rho_0}-S(A|K)_{\rho_2}\;,
    \end{equation}
    where the entropy of the post-measurement state of the target system conditioned by the classical register, $S(A|K)_{\rho_2}$, can equivalently be written as $\sum_kp_kS(\rho^A_{2,k})$, i.e., the average entropy of the posterior states of $A$. 
    \unskip\nobreak\hfill $\square$
\end{definition}

\begin{tech-remark}
    Note that $I_{\GO}$ is determined entirely by the prior system state $\rho^A_0$ and the instrument $\cA$ acting in $A$ as defined in \eq{eq:system-instrument}. The Groenewold--Ozawa information gain is guaranteed to be non-negative for all prior states $\rho_0^A$ if and only if the instrument $\cA$ is \emph{quasi-complete}; $\cA$ is called quasi-complete if for all pure prior states $\rho_0^A$, the posterior states $\rho^A_{2,k}:=\cA_k(\rho^A_0)/p_k$ are also pure. An example of a quasi-complete instrument is an \emph{efficient} instrument, whereby each operation can be written with a single Kraus operator, i.e., $\cA_k(\cdot) =  L_k (\cdot) L_k^\dagger$. In general, therefore,  $I_{\GO}$  can be negative~\cite{ozawa-1986-groen-info}.
\end{tech-remark}

The following proposition gives universally valid expressions for the work associated with feedback control and erasure protocols with a general quantum measurement  process, independent of thermodynamics and from a purely information-theoretic point of view.

\begin{proposition}\label{prop:general_formula}
    In the generalized quantum feedback control and erasure protocol (Fig.~\ref{fig:process}), the extracted work from the system is
    \begin{equation}\label{eq:extracable_holevo-go}
        W^A_\ext=-\Delta F^A_{0\to4}+\beta^{-1}\left[I_{\GO}-I(A\!:\!K)_{\rho_3} - S_\mathrm{irr}^{B_1} \right]\;, 
    \end{equation}
    and the work needed to run the controller is
    \begin{equation}\label{eq:work_in_go}
        W^{MK}_\In= \beta^{-1} \Delta S^{AMK}_{0\to2} + \beta^{-1}\left[I_{\GO} + I(A:M|K)_{\rho_2}+ S_\mathrm{irr}^{B_2}\right]\;,
    \end{equation}
where 
\begin{align*}
    S_\mathrm{irr}^{B_1} &:= \sum_{k\in \cK} p_k \bigg( I(A:B_1)_{\rho_{3,k}} + D(\rho^{B_1}_{3,k} \| \gamma^{B_1} )\bigg) \geqslant 0 , \nonumber \\
    S_\mathrm{irr}^{B_2} &:= I(MK\!:\!B_2)_{\rho_4}+D(\rho^{B_2}_4\|\gamma^{B_2}) \geqslant 0,
\end{align*}
denote the irreversible entropy production associated with the isothermal feedback and erasure steps. \unskip\nobreak\hfill $\square$
\end{proposition}

See Methods, Section~\ref{appendix:general_formula}, for the proof. We immediately see that \eq{eq:extracable_holevo-go} contains, besides the usual free energy change, a correction term that arises from the specific implementation of measurement and feedback protocol. Similarly, \eq{eq:work_in_go} contains additional correction terms to the usual entropy change of target system and controller.

We note that an equality similar to \eq{eq:extracable_holevo-go} was obtained in Ref. \cite{Jacobs2009}, except that there the entropy production  $S_\mathrm{irr}^{B_1}$ as well as the mutual information   $I(A:K)_{\rho_3}$ was missing. The term $I(A:K)_{\rho_3} = S(\rho_3^A) - \sum_{k \in \cK} p_k S(\rho^A_{3,k})$ corresponds to  the \textit{Holevo information} of the conditional states of $A$ after feedback~\cite{holevo1973bounds}, which is non-negative and vanishes if and only if $\rho^A_{3,k} = \rho^A_{3}$ for all $k$. Ref.~\cite{sagawa2012thermodynamics} also derives a similar equality, but it uses the QC-mutual information, and not the Groenewold--Ozawa information gain.

From \propref{prop:general_formula}, by discarding terms that are always either positive or negative,    we obtain universally valid  bounds for injected and extracted work in  quantum feedback control and erasure protocols, as well as necessary and sufficient conditions for their saturation.
\begin{framed}
\begin{theorem}\label{theorem:workbound}
In the generalized quantum feedback control and erasure protocol (Fig.~\ref{fig:process}), the work extracted from the target system is upper bounded as
\begin{equation}\label{eq:A-extractable}
    W^A_\ext\le-\Delta F^A_{0\to4}+\beta^{-1}I_{\GO}\;,
    \end{equation}
    where the equality holds if and only if $I(A\!:\!K)_{\rho_3}= S_\mathrm{irr}^{B_1} =0$.    The work cost to run the controller is lower bounded as
    \begin{equation}\label{eq:lower_demon}
    W^{MK}_\In\ge\beta^{-1}[\Delta S^{AMK}_{0\to2} + I_{\GO}]\;,
\end{equation}
where the equality holds if and only if $I(A\!:\!M|K)_{\rho_2} = S_\mathrm{irr}^{B_2} =0$.
\unskip\nobreak\hfill $\square$
\end{theorem}
\end{framed}

\begin{tech-remark}
Let us discuss, by means of examples, the conditions under which  the bounds in the above theorem can be saturated. A necessary condition for the equality in \eq{eq:A-extractable} is for the entropy production during the feedback step, $S_\mathrm{irr}^{B_1}$, to vanish. This will trivially be achieved if the feedback process is chosen to be  pure unitary, i.e.,  so that for each outcome the target system undergoes an isolated unitary evolution, as assumed in \cite{sagawa-ueda2008second}. However, note that in general this alone will not guarantee the other necessary condition for the equality in \eq{eq:A-extractable}, i.e., a vanishing Holevo information $I(A:K)_{\rho_3}$. Recall that this quantity vanishes if and only if $\rho^A_{3,k} = \rho^A_{3}$ for all $k$, which implies that $\rho^A_{3,k} = \rho^A_{3,k'}$ for all $k,k'$. But if the feedback process is pure unitary, then $\rho^A_{3,k} = F_k^A(\rho^A_{2,k})F_k^{A \dagger}$. Since unitary channels leave the von Neumann entropy invariant, and two states are identical only if their entropies are identical, it clearly follows that a necessary condition for a vanishing Holevo information given a pure unitary feedback process is for all the posterior states after measurement, $\rho^A_{2,k}$, to have the same entropy. While this can be achieved if, for example, the system undergoes a von Neumann measurement of a non-degenerate observable, for general measurement processes this is not  the case. This is why in physically relevant situations, in order to saturate \eq{eq:A-extractable} a feedback process that exchanges entropy with a thermal bath is required, thus going beyond the paradigm of pure unitary feedback processes employed in \cite{sagawa-ueda2008second}.
\end{tech-remark}

\begin{tech-remark}
Similarly as above, a necessary condition for the equality in \eq{eq:lower_demon} is for the entropy production during erasure, $S_\mathrm{irr}^{B_2}$, to vanish. The other necessary condition, however, is given by a vanishing conditional mutual information $I(A\!:\!M|K)_{\rho_2}=\sum_{k\in \cK} p_k \, I(A:M)_{\rho_{2,k}}$. Clearly, such a quantity vanishes if and only if $\rho^{AM}_{2,k} = \rho^A_{2,k} \otimes \rho^M_{2,k}$. Given that $\rho^{AM}_{2,k} = \idch^A \otimes \cM_k[ \cU(\rho^{A}_0 \otimes \rho^M_0)] / p_k$, a sufficient condition for $I(A\!:\!M|K)_{\rho_2}$ to vanish is if the instrument $\cM$ is \emph{nuclear} (also known as \textit{measure-and-prepare}~\cite{horodecki2003entanglement} or \textit{Gordon-Louisell type}~\cite{gordon-louisell}). That is,   if it holds that $\cM_k(\cdot) = \Tr[\cM_k(\cdot)] \varrho ^M_k$ for all $k$, where $\{\varrho ^M_k\}$ is a fixed family of states on $M$. It is clear that   a nuclear instrument acting in $M$ will destroy the correlations between $A$ and $M$ for each outcome $k$. Every POVM admits a nuclear instrument and, as shown in Corollary~1 of \cite{Heinosaari2010} (see also Theorem 2 of \cite{Pellonpaa2013a}), if the pointer observable measured by $\cM$ is rank-1, i.e., if all the effects $\M_k = \cM_k^*(\one^M)$ are proportional to a rank-1 projection, then $\cM$ is necessarily nuclear. Consequently, by choosing a rank-1 pointer observable, we can guarantee that the term $I(A\!:\!M|K)_{\rho_2}$ vanishes.
\end{tech-remark}

\subsection{Comparison between the second law of thermodynamics and the second law of ITh}

Our analysis so far has been independent of thermodynamics, but henceforth we will explore the consequences derived by combining the results of \propref{prop:general_formula} with the second law of thermodynamics.
Before doing so, however, we introduce two types of second laws of thermodynamics in this section, and show how they are related.

According to Ref.~\cite{esposito2011second}, when a thermodynamic system $Y$ transforms as $(\rho^Y_i; H^Y_i) \mapsto (\rho^Y_j; H^Y_j)$ by an isothermal processes, i.e., a process involving thermal baths with the same temperature, the second law can be formulated as the following inequality:
\begin{equation}\label{def:second_law}
    W_{\ext}^Y\le -\Delta F^Y_{i\to j}\;.
\end{equation}
Notice that the nonequilibrium free energy change in the right-hand side can be replaced by the change in \textit{equilibrium} free energy $F_{\mathrm{eq}}( H^Y):=-\beta^{-1}\ln Z^Y \equiv F(\gamma^Y; H^Y)$ whenever the initial state of $Y$ is assumed to be in thermal equilibrium---this is  a consequence of the implication~\cite{esposito2011second}
\begin{equation}\label{eq:eq-noneq}
 \rho^Y_i = \gamma^Y \implies   -\Delta F^Y_{i\to j}\le -\Delta F^Y_{\mathrm{eq},i\to j}.
\end{equation}
The above inequality will be useful when connecting our analysis to previous ones.

The  feedback control and erasure protocol we consider consists of the subsystem $AMK$  interacting with baths $B_1$ and $B_2$, which are assumed to be of the same temperature so that the total process is isothermal. This is to ensure that our analysis falls within the domain of applicability of the second  law as formulated in \eq{def:second_law}. As such, the feedback control and erasure protocol is consistent with the second law of phenomenological thermodynamics, or the \textit{overall} second law,  when the net extracted work given in \eq{eq:net-work-extracted} and the change in free energy of the compound $AMK$ obey the relation in  \eq{def:second_law}, i.e.,  
\begin{equation}\label{eq:clausius-like-feedback}
    W_{\textrm{tot}} \le-\Delta F^{AMK}_{0\to4}\;.
\end{equation}
We remark again that the above inequality embodies the second law of thermodynamics when considered from the beginning (time $t_0$) to the end (time $t_4$) of the protocol, regardless of what happens in the intermediate steps. On the other hand, the feedback control and erasure protocol is consistent with the second law of ITh, as formulated in~\cite{sagawa2009minimal},  when the net extracted work given in \eq{eq:net-work-extracted}  is bounded by the change in free energy of \emph{the target system alone}, i.e., 
\begin{align}\label{eq:new-second-law-IT-2}
    W_{\textrm{tot}} \le -\Delta F^A_{0\to4}\;.
\end{align}

Since the memory and register are erased, the free energy change $\Delta F^{MK}_{0\to4}$ is zero. Naively, one would be led to expect that $\Delta F^{AMK}_{0\to4} = \Delta F^{A}_{0\to4}$ as a result, thus suggesting that  Eqs. \eqref{eq:clausius-like-feedback} and \eqref{eq:new-second-law-IT-2} are equivalent. However, as the following proposition (proved in Methods, Section~\ref{appendix:second-law-equiv}) shows, they coincide \textit{if and only if} erasure is perfect. 

\begin{proposition}\label{prop:second-law-equiv}
The generalized quantum feedback control and erasure protocol (Figure~\ref{fig:process}) is consistent with the overall second law of thermodynamics,  i.e., \eq{eq:clausius-like-feedback}, if and only if
\begin{align}\label{eq:overall-second-law-equiv}
    \Delta S^{AMK}_{0\to2} &\geqslant   I(A:MK)_{\rho_4} - I(A:M|K)_{\rho_2} -I(A:K)_{\rho_3} -S_\mathrm{irr}^{B_1} - S_\mathrm{irr}^{B_2}. 
\end{align}
Instead, the protocol is consistent with the generalized second law of ITh, i.e., \eq{eq:new-second-law-IT-2}, if and only if 
\begin{align}\label{eq:second-law-ITh-equiv}
    \Delta S^{AMK}_{0\to2} &\geqslant    - I(A:M|K)_{\rho_2} -I(A:K)_{\rho_3} -S_\mathrm{irr}^{B_1} - S_\mathrm{irr}^{B_2}. 
\end{align}
Since $I(A:MK)_{\rho_4} \geqslant 0$, \eq{eq:overall-second-law-equiv} always implies \eq{eq:second-law-ITh-equiv}: consistency with the second law of thermodynamics implies consistency with the second law of ITh. The converse implication holds if and only if erasure is perfect, i.e., $\rho^{AMK}_4 = \rho^A_4 \otimes \rho^M_0 \otimes |0\>\<0|^K$.
\unskip\nobreak\hfill $\square$
\end{proposition}

In summary: a feedback control and erasure protocol that is consistent with the second law of phenomenological thermodynamics is guaranteed to also be consistent with the second law of ITh. However, if erasure is partial so that $I(A:MK)_{\rho_4} >0$, then it may be the case that the protocol is consistent with the second law of ITh, but violates the second law of thermodynamics proper, allowing for  work extraction beyond the Clausius bound.

\subsection{When is a quantum measurement process compatible with the second law?}

\propref{prop:second-law-equiv} above provides necessary and sufficient conditions for a given feedback control and erasure protocol to be consistent with the second law---be it the overall second law, or the second law of ITh. But now recall that a feedback control and erasure protocol is implemented by \textit{first} performing a measurement, and \textit{subsequently} performing feedback and erasure. It follows that, in order for a particular measurement process itself to be consistent with the second law(s), then all possible feedback control and erasure protocols that utilize that same measurement process must be consistent with the second law(s). This leads us to the following definition:

\begin{definition}
A given quantum measurement process 
   \begin{align*}
       \rho^{A}_0 \otimes \rho^M_0 \otimes |0\>\<0|^K \mapsto \sum_{k\in\cK}(\mathrm{id}^A\otimes\cM_k) \Big[\cU (\rho^{A}_0 \otimes \rho^M_0)\Big]\otimes\ket{k}\!\!\bra{k}^K
   \end{align*}
 is \textit{compatible with the overall second law of thermodynamic} whenever \eq{eq:overall-second-law-equiv}  holds for all possible subsequent isothermal feedback and erasure processes. Similarly, the measurement process is \textit{compatible with the second law of ITh} whenever \eq{eq:second-law-ITh-equiv} holds for all possible subsequent feedback and erasure processes.
  \unskip\nobreak\hfill $\square$
\end{definition}

We shall begin from a sufficient condition for a given measurement process to be compatible with the second law(s). As explicitly shown in Methods, Section~\ref{appendix:second-law-sufficient}, we observe that the right hand side of \eq{eq:overall-second-law-equiv} in \propref{prop:second-law-equiv} is never strictly positive, allowing us to  obtain the following:
\begin{proposition}
\label{prop:second-law-sufficient}
A   measurement process that does not decrease the total entropy, i.e., such that $\Delta S^{AMK}_{0\to2}\ge 0$, is guaranteed to be compatible with the overall second law and, hence, also with the second law of ITh. Moreover, a sufficient condition for $\Delta S^{AMK}_{0\to2}\ge 0$ to hold is if the instrument $\cM$ responsible for pointer objectification implements a bistochastic channel, i.e., a CP linear map that preserves both the trace and the unit.  
\unskip\nobreak\hfill $\square$
\end{proposition}

A consequence of \propref{prop:second-law-sufficient} is that a feedback control and erasure protocol may violate the second law(s) only if it includes a measurement process that decreases the total entropy.
However,  it does not follow that \emph{any} measuring process that decreases the entropy will \textit{always} violate the second law(s). To this end, we obtain the following necessary condition for a measurement process to be compatible with the second law(s), proven in Methods, Section~\ref{appendix:second-law-necessary}:

\begin{framed}
\begin{theorem}\label{theorem:second-law-necessary}
The  measurement process is compatible with the  second law of ITh if and only if
\begin{align}\label{eq:second-law-ITh-necessary-1}
 \Delta S^{AMK}_{0\to2} \geqslant - I(A:M|K)_{\rho_2},   
\end{align}
or, equivalently, 
\begin{align}\label{eq:second-law-ITh-necessary-2}
    \mathscr{H}(\{ p_k\}) \geqslant I_{\GO} + J_{\GO},
\end{align}
where $\mathscr{H}(\{p_k\}) := -\sum_{k\in \cK} p_k \ln p_k$ is the Shannon entropy of the measurement outcomes probability distribution, and $J_{\GO} := S(M)_{\rho_0} - S(M|K)_{\rho_2}$ is the Groenewold--Ozawa information gain of the memory.

Moreover, the above inequalities are necessary conditions for the measurement process to be compatible with the overall second law.
\unskip\nobreak\hfill $\square$
\end{theorem}
\end{framed}

\eq{eq:second-law-ITh-necessary-1} states that even if the measurement process decreases the entropy, as long as the target system and memory are left in a sufficiently correlated state, then all possible feedback and erasure processes built on it will still be consistent with the second law of ITh. Such a condition is equivalently reformulated in \eq{eq:second-law-ITh-necessary-2} as a \textit{tradeoff} between the information gains of the target system and the memory: if a given measurement process is compatible with the second law of ITh, then the information gain of the target system and that of the memory cannot be both arbitrarily large at the same time, but their sum must remain below the Shannon entropy of the measurement outcomes distribution.

Note that, in Theorem~\ref{theorem:second-law-necessary}, it is the entropy of the \textit{compound} $AMK$ that matters, not the entropy of the system $A$ alone, which may well decrease as a result of the action of the effective instrument $\{\mathcal{A}_k:k\in\mathcal{K}\}$ in Eq.~\eqref{eq:system-instrument}. In other words, the second law puts a restriction \textit{on how a particular instrument is realized on the compound}, not on the instrument itself.

\begin{tech-remark}
As stated in \propref{prop:second-law-sufficient},  if the instrument responsible for pointer objectification implements a bistochastic channel,  then the entropy of the compound $AMK$ is guaranteed not to decrease~\cite{Alberti1982}, thereby ensuring compatibility with the second law. A paradigmatic example of an objectification process that satisfies this condition is given by the L\"uders instrument.  But for any pointer observable $\M$ acting in the memory, there are $\M$-compatible instruments which do not implement a bistochastic channel---for example, a nuclear instrument which prepares the memory in the same pure state for all outcomes. Additionally, let us recall that such an instrument will always destroy the correlations between system and memory, so that $I(A:M|K)_{\rho_2} = 0$, whereby a decrease in entropy is sufficient for the violation of the second law for some feedback and erasure process. In  Methods, Section~\ref{appendix:bistochastic}, we explicitly construct such a feedback control and erasure protocol so that $\Delta S^{AMK}_{0\to2}$ is strictly negative, and which violates both the overall second law of thermodynamics, as well as the second law of ITh. 
\end{tech-remark}

As a consequence of the above, we see that the choice of the measurement process, in particular, of the objectification process, while not affecting the dynamics of the target system alone---which depends only on the pointer observable $\M$, not on the choice of $\cM$ implementing it, see \eq{eq:system-instrument}---instead has a \textit{non-trivial thermodynamic implication}, since the state change of the memory enters directly into the thermodynamic balance. In fact, the common assumption that the pointer objectification is implemented by a L\"uders instrument~\cite{sagawa2009minimal,sagawa2009erratum, Jacobs2009,funo2013integral,abdelkhalek2016fundamental,strasberg17,Mohammady2019c, Strasberg2019, Strasberg2020b,  strasberg2022quantum, Latune2024} obscures the role that the bistochasticity of such instruments plays in ensuring consistency with the second law, leading to the erroneous conclusion that the laws of quantum theory alone are sufficient to ensure compatibility with the second law. Here instead we have shown that, in order to obtain a full understanding of how the pointer objectification relates to the second law, the instrument $\cM$ \textit{must be treated as arbitrary}, as we have done, \textit{lest one obtain statements of limited scope}.

\section{Discussion}

Here, we compare the work inequalities presented in Theorem~\ref{theorem:workbound} with those previously obtained by Sagawa and Ueda~\cite{sagawa-ueda2008second,sagawa2009minimal,sagawa2009erratum}. According to \cite{sagawa-ueda2008second}, the achievable upper bound on the amount of work extracted by feedback control from the target system $A$, assumed to be initially in equilibrium, is
\begin{equation}\label{eq:sagawa-ueda-upper}
    W^A_\ext\le-\Delta F^A_{\Eq,\;0\to4}+\beta^{-1}I_{\QC}\;,
\end{equation}
where $I_{\QC}$ is a nonnegative quantity named the \textit{QC-mutual information}~\cite{sagawa-ueda2008second}. This quantity, in some particular situations, can be interpreted as a measure of the information gained by the measurement performed by the controller on the target system. Thus
Eq.~\eqref{eq:sagawa-ueda-upper} implies that the second law~\eqref{def:second_law} for system $A$ can be violated in a feedback control protocol by an amount that is directly proportional to the information that the controller is able to obtain about the target system. Then, in a subsequent paper~\cite{sagawa2009minimal}, the same authors showed that the quantity $\beta^{-1}I_{\QC}$, under suitable assumptions, provides a tight  lower bound on the work cost for  measurement and  erasure:
\begin{equation}\label{eq:sagawa-ueda-lower}
    W^{MK}_\meas+W^{MK}_\eras\equiv W^{MK}_\In\ge \beta^{-1}I_{\QC} \;.
\end{equation}
Recalling that $ W_{\textrm{tot}} = W^A_\ext-W^{MK}_\In$,  one thus obtains
\begin{equation}
  W_{\textrm{tot}} \le-\Delta F^A_{\Eq,\;0\to4}\;,\label{eq:sagawa-ueda-2ndlaw}
\end{equation}
which Ref.~\cite{sagawa2009minimal} refers to as the second law of ITh.

However, in order to be valid, the analysis presented by Sagawa and Ueda in~\cite{sagawa-ueda2008second,sagawa2009minimal,sagawa2009erratum} requires the following assumptions on the quantum feedback control and erasure protocol:

\begin{description}
    \item[\textbf{Assumption 1 (A-1)}~\cite{sagawa2009minimal}] The pointer objectification must be implemented by a L\"uders instrument $\cM^L_k(\cdot) := \M_k (\cdot) \M_k$  compatible with a projection valued measure $\M$ acting in $M$. That is,  for each measurement outcome $k$, it must hold that
        \begin{equation*}
            \rho^{AM}_{2,k}=\frac{(\mathds{1}^A\otimes \M_k) \cU(\rho^A_0\otimes\rho^M_0) (\mathds{1}^A\otimes \M_k)}{p_k}\;.
        \end{equation*}

    \item[\textbf{Assumption 2 (A-2)}~\cite{sagawa-ueda2008second,sagawa2009minimal,sagawa2009erratum}] The instrument acting in the target system $A$, i.e., $\cA_k(\cdot) :=\Tr_M\big\{(\mathrm{id}^A\otimes\cM_k) \big[\cU(\cdot\otimes\rho^M_0)\big]\big\}$, must be \emph{efficient}. That is, every operation $\cA_k$ must be expressible with only one Kraus operator. 

    \item[\textbf{Assumption 3 (A-3)}~\cite{sagawa-ueda2008second}] The target system $A$ must be initially prepared in the Gibbs state, that is, $\rho^A_0=\gamma^A$.

    \item[\textbf{Assumption 4 (A-4)}~\cite{sagawa2009minimal}] At time step $t=t_2$, the target system and memory must be in a product state for each outcome $k$, i.e., $\rho^{AM}_{2,k} = \rho^A_{2,k}\otimes\rho^M_{2,k}$.

    \item[\textbf{Assumption 5 (A-5)}~\cite{sagawa-ueda2008second}] The feedback process must be pure unitary. That is, for each outcome $k$ it must hold that $\rho^A_{3,k} = F^A_k (\rho^A_{2,k}) F^{A\dagger}_k$.

    \item[\textbf{Assumption 6 (A-6)}~\cite{sagawa2009minimal}] The memory's Hilbert space and Hamiltonian possess a  direct sum structure, i.e., $\cH^M = \bigoplus_{k=0}^N \cH^{M_k}$ and   $H^M = \bigoplus_{k=0}^N H^{M_k}$, where $N = |\cK|$ is the number of measurement outcomes, and $H^{M_k}$ are Hamiltonians on the sector $\cH^{M_k}$. Denoting the Gibbs states for each sector $\cH^{M_k}$ as $\gamma^{M_k}$, it must hold that: (i) the initial state of the memory satisfies $\rho^M_0 = \gamma^{M_0}$, and (ii)  the conditional states of the memory \textit{before} erasure are thermal in the respective sectors, i.e., $\rho^M_{3,k} =\gamma^{M_k}$.

\end{description}

Note that none of the above assumptions need be satisfied by a general measurement and feedback process like that we consider. In fact, they are generally incompatible, except in trivial cases, as we discuss in the following remark.

\begin{tech-remark}
First, assumptions (A-1) and (A-4) are typically incompatible, since given a L\"uders-type pointer objectification,  the post-measurement states $\rho^{AM}_{2,k}$ will in general be correlated. There are two cases in which  (A-4) will be guaranteed to hold given (A-1): (i) if   $\M_k$ are rank-1 projections, which is both necessary and sufficient for the $\M$-compatible L\"uders instrument $\cM^L$ to be nuclear, then measurement of $M$ by $\cM^L$ is guaranteed to destroy the correlations between $A$ and $M$; (ii) if the premeasurement unitary channel is local, i.e., $\cU = \cU^A \otimes \cU^M$, then it trivially holds that $\rho^{AM}_{2,k} = p_k^{-1} \cU^A(\rho^A_0) \otimes \M_k \cU^M(\rho^M_0) \M_k$.  But in such a case the measurement process does not extract any information at all, as it implements a trivial observable in $A$, namely, a POVM whose elements are all proportional to $\one^A$. Second, whenever the elements of the POVM measured by the instrument $\cA$ in the target system are linearly independent (for example, if the observable is projection valued) then (A-1), (A-2), and (A-6) are compatible  only if 
$\dim(\cH^{M_0}) \leqslant N^{-1}\sum_{k=1}^N \dim(\cH^{M_k})$.
This follows from the fact that Gibbs states have full rank, and so the rank of $\rho^M_0 = \gamma^{M_0}$ equals $\dim(\cH^{M_0})$, together with the fact that an efficient instrument compatible with an observable with linearly independent effects is \emph{extremal} \cite{DAriano2011,Pellonpaa2013}.  See  Methods, Section~\ref{appendix:efficient-instrument-rank}, for the proof. In particular, since $\M_k$ are projections onto the subspaces $\cH^{M_k}$, then if $\M_k$ are rank-1 projections, which is necessary to guarantee compatibility of (A-1) and (A-4) discussed above, then $\cH^{M_0}$ must also be 1-dimensional. In other words, in order to guarantee compatibility between assumptions (A-1), (A-2), (A-4), and (A-6), the initial state of the memory, $\rho^M_0$,  must be pure. This is a physically unrealistic assumption due to the third law of thermodynamics \cite{Taranto2021}. 
\end{tech-remark}

On the other hand, as a consequence of our analysis, one easily sees that in fact Assumption~(A-1) alone is already sufficient to obtain Eq.~\eqref{eq:new-second-law-IT-2} which, under Assumption (A-3) and \eq{eq:eq-noneq}, directly implies Eq.~\eqref{eq:sagawa-ueda-2ndlaw}.
This is because  L\"uders channels  are bistochastic, so that  by \propref{prop:second-law-sufficient}  $\Delta S^{AMK}_{0\to2}\ge 0$ is guaranteed to hold, which implies consistency with both second laws. 

Thus, Eqs.~\eqref{eq:A-extractable} and \eqref{eq:lower_demon} constitute a strict extension of Sagawa and Ueda's relations~\eqref{eq:sagawa-ueda-upper} and~\eqref{eq:sagawa-ueda-lower}.
This is because:
\begin{enumerate}
\setlength{\itemsep}{-2mm}
    \item When the pointer objectification is implemented by a projective measurement on the memory, i.e., under (A-1), it holds that $\Delta S^{AMK}_{0\to2}\ge 0$. Moreover, if $\Delta S^{AMK}_{0\to2}>0$, Eq.~\eqref{eq:lower_demon} is a more refined inequality than Eq.~\eqref{eq:sagawa-ueda-lower}, and while the former can be saturated, the latter cannot.
  
    \item When the instrument acting in $A$ is assumed to be efficient, i.e., under (A-2), then the Groenewold--Ozawa information gain $I_{\operatorname{GO}}$ coincides with the QC-mutual information $I_{\operatorname{QC}}$, as shown in Ref.~\cite{buscemi2008global}; in all other situations, the two quantities are unrelated, i.e., $I_{\GO}\gtreqqless I_{\QC}$, but the one that retains its role in thermodynamic relations is $I_{\GO}$.

    \item When the target system is initialized in a Gibbs state, i.e., under (A-3), then $-\Delta F^A_{0\to4}\le -\Delta F^A_{\Eq,0\to4}$ because of Eq.~\eqref{eq:eq-noneq}.
\end{enumerate}
In particular, we conclude that the correct information measure that remains valid for general measurement processes is $I_{\GO}$, not $I_{\QC}$.
Although  $I_{\GO}$ has been considered also in some previous works~\cite{Jacobs2009, funo2013integral,abdelkhalek2016fundamental}, these still imposed assumption (A-1). Our analysis shows that $I_{\GO}$ is the right quantity to consider even when (A-1) is not satisfied.

Summarizing, in this paper, we have shown that  the consistency between the second law of thermodynamics and information processing is not guaranteed by the laws of quantum theory \emph{simpliciter}. Instead, the second law must be taken as a primitive principle which imposes constraints on the physically valid quantum information processing protocols. In order to precisely characterize such constraints, we formulated quantum feedback control and erasure protocols with general isothermal feedback  and general measurement processes. In particular, we did not assume that the  pointer objectification step of the measurement process is implemented by a L\"uders instrument, as was done in previous studies. We then provided  necessary and sufficient conditions for such protocols to be consistent with the second law  (\propref{prop:second-law-equiv}). More generally, we provided necessary and sufficient conditions for a given measurement process to be consistent with the second laws for all subsequent feedback control and erasure processes (\propref{prop:second-law-sufficient} and \thmref{theorem:second-law-necessary}). These results show that while the second law is necessarily obeyed if the pointer objectification process is bistochastic---as is the case for L\"uders instruments---the second law can be violated if the pointer objectification decreases the entropy, which is permitted by quantum theory alone. In this very sense, then, quantum theory alone is not a guarantee of compatibility with the second law.

Along the way, we derived expressions for the work extracted by feedback control and the work required for measurement and erasure (\propref{prop:general_formula} and \thmref{theorem:workbound}) which, 
unlike those presented in previous studies~\cite{sagawa-ueda2008second,sagawa2009minimal,sagawa2009erratum, Jacobs2009,funo2013integral,abdelkhalek2016fundamental,Mohammady2019c, Strasberg2019, Strasberg2020b},  are universally valid in the sense that we did not impose any assumptions on the feedback process, the measurement (including the pointer readout), or the initial state of the system. Of course, our equations recover those presented in previous studies~\cite{sagawa-ueda2008second,sagawa2009minimal,sagawa2009erratum}, but are able to do  so with fewer assumptions. As our other main result, we then show that the generalized second law of ITh presented here is guaranteed to hold  for any quantum feedback control and erasure protocol that is consistent with the second law of  thermodynamics proper, and that the two laws become equivalent  in the case of perfect erasure of the demon's memory (\propref{prop:second-law-equiv}).

This resolves the problem of the scope of the second law of ITh, which was unclear from previous studies, but can now be considered a \textit{universally valid law of physics}. That is to say, since the conjunction of the second law and the laws of quantum theory implies that the second law of ITh will hold by logical necessity, as long as the second law and quantum theory are regarded as universally valid laws of physics, then so too must the second law of ITh be.
Our results also contribute to the debate regarding the operational interpretation of the Groenewold--Ozawa information gain, which has been generally considered problematic, especially in those situations where it takes negative values; we have seen that this quantifies the amount by which the extractable work by measurement-plus-feedback exceeds the reduction in free energy \cite{Jacobs2009,Strasberg2019}, for all possible measurement and feedback processes.

An interesting direction to follow will be to look for applications of our approach to other formulations of the second law such as fluctuation theorems~\cite{sagawa2012fluctuation,sagawa2012nonequilibrium,sagawa2013role,funo2013integral,buscemi-scarani-2021fluctuation,aw-buscemi-scarani}.
In the same way, another possible line for future research is to bring our analysis to the one-shot case~\cite{horodecki2013fundamental,faist2018fundamental,Lipka-Bartosik2018}, possibly beyond quantum theory~\cite{hanggi2013violation,krumm2017thermodynamics,minagawa-2022}, and to introduce insights from the thermodynamic reverse bound~\cite{buscemi2020thermodynamic}, retrodiction~\cite{buscemi-scarani-2021fluctuation,aw-buscemi-scarani,Buscemi-Safranek-Schindler-OE-NJP-2023,bai2023observational} and the theory of approximate recoverability~\cite{buscemi-das-wilde}. Finally, an interesting line of future investigation will be to see how the second law of ITh interplays with the first and third laws of thermodynamics: the first law demands that the interaction between system and memory of the measuring device must be constrained so as to conserve the total energy, whereby the Wigner--Araki--Yanase theorem will impose limitations on the measurements one may perform \cite{Wigner:1952aa,Araki1960,Ozawa2002,tajima2019coherence,Mohammady2021a,Kuramochi2022,emori2023error}. On the other hand, the third law will prohibit the memory from being initialized in a pure state, which has also been shown to impose fundamental constraints on measurements \cite{Guryanova2018,Mohammady2022a, Mohammady2024}. While we have seen that the second law alone imposes no constraints on the measurements we can make on the target system---any instrument acting in the target system allows for a bistochastic measurement process that does not reduce the total entropy of the compound---it may be the case that, in conjunction with the other laws of thermodynamics, further constraints must be imposed on the quantum measurements that can be performed.

\section{Methods}

\subsection{Preliminaries}\label{app:preliminaries}

Here, we introduce some preliminary concepts  which will be used in the technical proofs appearing throughout the rest of the manuscript.

\begin{definition}\label{defn:free-energy}
 Consider a thermodynamic system $(\rho^A; H^A)$. The \emph{internal energy} is defined as
 \begin{align*}
 E(\rho^A;H^A):=\Tr[\rho^A H^A],    
 \end{align*}
 and the \emph{nonequilibrium free energy}~\cite{gaveau1997general,esposito2011second}  is defined as 
    \begin{equation*} 
    F(\rho^A; H^A):= E(\rho^A; H^A)-\beta^{-1}S(A)_\rho \equiv F_{\mathrm{eq}}( H^A)+\beta^{-1}D(\rho^A\|\gamma^A)\;,
    \end{equation*}
    where  $F_{\mathrm{eq}}( H^A):=-\beta^{-1}\ln Z^A \equiv F(\gamma^A; H^A)$ is the equilibrium (Helmholtz) free energy.
    \unskip\nobreak\hfill $\square$
\end{definition}

\begin{lemma}\label{lemma:free_energy_mutual}
Consider a bipartite thermodynamic system $(\rho^{AB}; H^{AB})$. Assume that the Hamiltonian is additive, i.e., $ H^{AB} = H^A + H^{B}:= H^A\otimes\mathds{1}^{B}+\mathds{1}^A\otimes H^{B}$. It holds that
\begin{align*}
 E(\rho^{A B}; H^{AB})= E(\rho^A; H^A)+E(\rho^{B};H^{B})   
\end{align*}
and
\begin{equation*}
F(\rho^{AB}; H^{AB})=F(\rho^A; H^A)+F(\rho^{B}; H^{B})+\beta^{-1}I(A\!:\!B)_\rho\;.
\end{equation*}
\unskip\nobreak\hfill $\square$
\end{lemma}
\begin{proof}
Note that by the definition of the partial trace, it holds that $\Tr[ \rho^{AB} L^A \otimes \one^B] = \Tr[\rho^A L^A ]$ for all $L^A$ and $\rho^{AB}$. The additivity of the internal energy follows trivially from the additivity of the Hamiltonian. Now note that $F(\rho^{AB};H^{AB})=E(\rho^{A B};H^{AB}) -\beta^{-1}S(A B)_{\rho}$. Observing that $S(A B)_{\rho} = S(A)_{\rho}+S(B)_{\rho}-I(A:B)_{\rho}$ completes the proof. 
\end{proof}

\emph{Operations} provide the most general description for how a quantum system may transform. In the Schr\"odinger picture, an operation acting in a system $A$ is defined as a completely positive (CP), trace non-increasing linear map $\Phi : \cL(\cH^A) \to \cL(\cH^A)$. We shall denote the consecutive application of operations $\Phi_1$ followed by $\Phi_2$ as $\Phi_2 \circ \Phi_1$. For each operation, there exists a Heisenberg picture dual $\Phi^*$, defined by the trace duality $\Tr[ \Phi^*(L^A) \rho^A ] = \Tr[ L^A \Phi(\rho^A) ]$ for all $\rho^A$ and $L^A$. $\Phi^*$ is a  sub-unital CP  linear map, i.e., $\Phi^*(\one^A) \leqslant \one^A$. Among the operations are \emph{channels}, which preserve the trace, and if $\Phi$ is a channel, then $\Phi^*$ is unital, i.e., $\Phi^*(\one^A) = \one ^A$. We shall denote the identity channel acting in $A$ as $\idch^A$, which satisfies $\idch^A(L^A) = L^A$ for all $L^A$.  An operation acting in a composite system $AB$ is \emph{local} if it can be written as $\Phi = \Phi^A \otimes \Phi^B$, such that $\Phi(L^A \otimes L^B) = \Phi^A(L^A) \otimes \Phi^B(L^B)$ for all $L^A$ and $L^B$. As such, $\Phi^A \otimes \idch^B$ is an operation that acts locally and non-trivially only in subsystem $A$.

\begin{lemma}\label{lemma:local-energy-change}
 Consider a bipartite thermodynamic system which transforms as $(\rho^{AB}_i; H_i^{AB}) \mapsto (\rho^{AB}_j; H_j^{AB})$, such that $\rho^{AB}_j =  \Phi^A \otimes \idch^B(\rho^{AB}_i)$, where $\Phi^A$ is a channel acting in $A$ and $\idch^B$ is the identity channel acting in $B$. The following hold:
 \begin{enumerate}[(i)]
     \item $\rho^A_j = \Phi^A(\rho^A_i)$ and $\rho^B_j = \rho^B_i$.
     \item If $H_k^{AB} = H^A_k + H^B$ for $k=i,j$, then $\Delta E^{AB}_{i \to j}  = \Delta E ^A_{i \to j } = \Tr[\Phi^A(\rho^A_i) H^A_j] - \Tr[\rho^A_i H^A_i]$. 
 \end{enumerate}
 \unskip\nobreak\hfill $\square$
\end{lemma}
\begin{proof}
 \begin{enumerate}[(i):]
 \item For all $L^A$ and $L^B$, it holds that
\begin{align*}
    \Tr[\rho^A_j L^A ] &= \Tr[ \Phi^A \otimes \idch^B(\rho^{AB}_i) (L^A \otimes \one^B)] = \Tr[\rho^{AB}_i \Phi^{A*} \otimes \idch^B(L^A \otimes \one^B)] \\
    & = \Tr[\rho^{AB}_i \Phi^{A*}(L^A) \otimes  \one^B ] = \Tr[ \rho^{A}_i \Phi^{A*}(L^A)] = \Tr[\Phi^A(\rho^A_i) L^A ], \\
    \Tr[\rho^B_j L^B ] & = \Tr[\Phi^A \otimes \idch^B(\rho^{AB}_i) (\one^A \otimes L^B) ] = \Tr[\rho^{AB}_i \Phi^{A*} \otimes \idch^B(\one^A \otimes L^B) ] \\
    & = \Tr[\rho^{AB}_i \one^A \otimes L^B ] = \Tr[\rho^B_i L^B ].
\end{align*}
Here, we have used the definition of the partial trace, the trace duality, and the fact that $\Phi^{A*}$ is unital while $\idch^B(L^B) = L^B$ for all $L^B$. Since $\Tr[\rho^A L^A ] = \Tr[\sigma^A L^A ]$ for all $L^A$ if and only if $\rho^A = \sigma^A$ completes the proof. 
\item  This follows from item (i), together with the additivity of the Hamiltonian, Lemma \ref{lemma:free_energy_mutual}, and the fact that $H^B_i = H^B_j = H^B$. 
 \end{enumerate}
\end{proof}

\begin{lemma}\label{lemma:unitary-isothermal-work}
    Consider a system $Y$ and a thermal bath $B$, which transform as $(\rho^{YB}_i; H^{YB}_i) \mapsto (\rho^{YB}_j; H^{YB}_j)$. Assume that $\rho^{YB}_i := \rho^Y_i \otimes \gamma^B$, and that  $\rho^{YB}_j = \Phi (\rho^{YB}_i)$ with $\Phi(\cdot) := U(\cdot) U^\dagger$ a unitary channel, and that   $H^{YB}_k := H^Y_k + H^B$ for $k=i,j$.  Then the extracted work from system $Y$ will read
   \begin{align*}
       W_\mathrm{ext}^Y = -\Delta E^{YB}_{i \to j} = -\Delta F^Y_{i \to j} - \beta^{-1} S_\mathrm{irr}^B, 
   \end{align*} 
where 
\begin{align*}
 S_\mathrm{irr}^B :=   I(Y:B)_{\rho_j} +  D( \rho^B_j \| \gamma^B)  \geqslant 0  
\end{align*}
is the irreversible entropy production, vanishing if and only if $\rho^{YB}_j = \rho^Y_j \otimes \gamma^B$.  
\unskip\nobreak\hfill $\square$
\end{lemma}

\begin{proof}
Since unitary evolution is adiabatic, then by Definition \ref{defn:work-adiabatic} the extracted work from the compound $YB$ will equal the decrease in internal energy, and so by Definition \ref{defn:free-energy} it holds that  $W_\mathrm{ext}^{YB} := - \Delta E^{YB}_{i \to j} = -\Delta F^{YB}_{i \to j} - \beta^{-1} \Delta S^{YB}_{i \to j} = -\Delta F^{YB}_{i \to j}$, with the last step following from the fact that unitary evolution does not change the von Neumann entropy. Now note that by the first law of thermodynamics, it holds that $W^Y_\ext = - \Delta E^Y_{i \to j} - Q^Y$, where $W^Y_\ext$ is the work extracted from system $Y$, and $Q^Y := \Delta E^B_{i \to j}$ is the heat that flows  to the bath $B$. By  the additivity of the Hamiltonian and Lemma \ref{lemma:free_energy_mutual}, it  follows that $W^Y_\ext = - \Delta E^Y_{i \to j} - \Delta E^B_{i \to j} = - \Delta E^{YB}_{i \to j} =: W_\mathrm{ext}^{YB}$. We may therefore write
\begin{align*}
 W_\mathrm{ext}^Y &= -\Delta F^{YB}_{i \to j}\\
 & = -\Delta F^{Y}_{i \to j} - \Delta F^{B}_{i \to j} -  \beta^{-1} I(Y:B)_{\rho_j}  \\
 & = -\Delta F^{Y}_{i \to j}   - \beta^{-1} \left[    I(Y:B)_{\rho_j} + D( \rho^B_j \| \gamma^B) \right].
\end{align*}
In the second line we have used Lemma \ref{lemma:free_energy_mutual} and the additivity of the Hamiltonian, together with the fact that system and bath are uncorrelated at initial time, and so $I(Y:B)_{\rho_i} = 0$. In the third line we use the fact that the bath is initially in thermal equilibrium, i.e., $\rho^{B}_i = \gamma^B$, together with Definition \ref{defn:free-energy} and the fact that the bath Hamiltonian, and hence the bath equilibrium free energy, does not change. Finally, we recall that the mutual information $I(Y:B)_{\rho_j}$ is non-negative and vanishes if and only if $\rho^{YB}_j = \rho^Y_j \otimes \rho^B_j$, whereas the relative entropy $D( \rho^B_j \| \gamma^B)$ is non-negative and vanishes if and only if $\rho^B_j = \gamma^B$~\cite{nielsen_chuang_2010}.
\end{proof}

\subsection{Proof of Proposition~\ref{prop:general_formula}}
\label{appendix:general_formula}

We shall first prove \eq{eq:extracable_holevo-go}. Given that feedback is implemented by a global unitary channel $\rho_2 \mapsto \rho_3 = \cF \otimes \idch ^{M B_2}(\rho_2)$, the extracted work will read 
\begin{align}\label{eq:app-feedback-work-1}
 W_{\ext,2 \to 3}^A &:= -\Delta E_{2 \to 3}  = \Tr[\rho_2 \, H_2] - \Tr[\cF \otimes \idch ^{M B_2}(\rho_2) H_3] \nonumber  \\
 & = \Tr[\rho^{B_1 A}_2 (H^{B_1} + H^A_2)] - \Tr[\rho^{B_1 A}_3(H^{B_1} + H^A_3)] \nonumber \\
 & = \sum_{k\in \cK} p_k \bigg(\Tr[ \gamma^{B_1} \otimes \rho_{2,k}^{A} \,( H^{B_1} + H^A_2)] - \Tr[ \cF_k (\gamma^{B_1} \otimes \rho_{2,k}^{A}) (H^{B_1} + H^A_3)] \bigg) \nonumber \\
 & = - \sum_{k\in \cK} p_k \bigg( \Delta F^{A}_{2 \to 3, k} + \beta^{-1} [ I(A:B_1)_{\rho_{3,k} } + D( \rho^{B_1}_{3,k} \| \gamma^{B_1})] \bigg).
\end{align}
Here, the second line follows from Lemma \ref{lemma:local-energy-change} and the fact that $\cF$ acts locally in $B_1 A K$,  and that the Hamiltonian at $t_2, t_3$ is additive with only the Hamiltonian of $A$ changing in time, and that the state of $K$ does not change. The third line follows from \eq{eq:state-after-measurement} and \eq{eq:state-after-feedback}. The final line follows from Lemma \ref{lemma:unitary-isothermal-work}.   Now let us note that we may write
\begin{align}\label{eq:app-feedback-work-2}
- \sum_{k\in \cK} p_k   \Delta F^{A}_{2 \to 3, k} &= \sum_{k\in \cK} p_k \bigg(\Tr[\rho^A_{2,k} H^A_2 ] - \Tr[\rho^A_{3,k} H^A_3 ] + \beta^{-1}[S(\rho^A_{3,k}) - S(\rho^A_{2,k})] \bigg)  \nonumber \\
& = \Tr[\rho ^A_2 H^A_2 ] - \Tr[\rho^A_{3} H^A_3 ] + \beta^{-1} \sum_{k \in \cK} p_k [S(\rho^A_{3,k}) - S(\rho^A_{2,k})] \nonumber \\
& = (\Tr[\rho ^A_2 H^A_2 ] - \Tr[\rho^A_0 H^A_0 ]) + (\Tr[\rho^A_0 H^A_0 ] - \Tr[\rho^A_{3} H^A_3 ]) \nonumber \\
& \qquad + \beta^{-1}  [I_{\GO} + S(\rho_3^A) - S(\rho_0^A) - I(A:K)_{\rho_3 }] \nonumber \\
& = \Delta E^A_{0 \to 2} -\Delta F^A_{0 \to 3} + \beta^{-1}  [I_{\GO} - I(A:K)_{\rho_3 }] \nonumber \\
& = \Delta E^A_{0 \to 2} -\Delta F^A_{0 \to 4} + \beta^{-1}  [I_{\GO} - I(A:K)_{\rho_3 }].   
\end{align}
In the second line we use the fact that $\sum_{k\in \cK} p_k \, \rho^A_{i, k} = \rho^A_i$. The third line is obtained by adding and subtracting $\Tr[\rho^A_0 H^A_0 ]$ , $\beta^{-1}   S(\rho_0^A) $, and $\beta^{-1}   S(\rho_3^A) $, and noting that $I_{\GO} =  S(\rho_0^A) - \sum_{k\in \cK} p_k S(\rho^A_{2,k})$ and $I(A:K)_{\rho_3} =  S(\rho_3^A) - \sum_{k\in \cK} p_k S(\rho^A_{3,k})$. The final line is obtained by noting that $\Delta F^A_{0 \to 4} = \Delta F^A_{0 \to 3} + \Delta F^A_{3 \to 4}$, and that $\Delta F^A_{3 \to 4}=0$ since both the state and Hamiltonian of system $A$ do not change between time step $t_3$ and $t_4$. Finally, since $W^A_{\ext, 0 \to 2} = - \Delta E^A_{0 \to 2}$, then by \eq{eq:app-feedback-work-1} and \eq{eq:app-feedback-work-2} we have that 
\begin{align*}
  W_\ext ^A &=   W^A_{\ext, 0 \to 2} + W_{\ext,2 \to 3}^A \\
  & = -\Delta F^A_{0 \to 4} + \beta^{-1}  \bigg( I_{\GO} - I(A:K)_{\rho_3 } - \sum_{k \in \cK} p_k [ I(A:B_1)_{\rho_{3,k} } + D( \rho^{B_1}_{3,k} \| \gamma^{B_1})] \bigg),
\end{align*}
and so we obtain \eq{eq:extracable_holevo-go}.

Next, we show Eq.~\eqref{eq:work_in_go}. Since the erasure step is implemented by the global unitary channel $\rho_3 \mapsto \rho_4 = \idch^{B_1 A} \otimes \cV(\rho_3)$, we have 
\begin{align*}
    W_{\In, 3 \to 4}^{MK} &:= \Delta E_{3 \to 4} = \Tr[\idch^{B_1 A} \otimes \cV(\rho_3) H_4] - \Tr[\rho_3 H_3] \nonumber \\
    & = \Tr[\cV(\rho_3^{MK} \otimes \gamma^{B_2}) (H^{MK} + H^{B_2})] - \Tr[\rho_3^{MK} \otimes \gamma^{B_2} (H^{MK} + H^{B_2})] \nonumber \\
    & = \Delta F^{MK}_{3 \to 4} + \beta^{-1} [I(MK:B_2)_{\rho_4} + D(\rho_4 ^{B_2} \| \gamma^{B_2})] \nonumber \\
    & = -\Delta F^{MK}_{0 \to 2} + \beta^{-1} [I(MK:B_2)_{\rho_4} + D(\rho_4 ^{B_2} \| \gamma^{B_2})].
\end{align*}
The second line follows from Lemma \ref{lemma:local-energy-change} and the fact that $\cV$ acts locally in $MK B_2$, and the fact that the Hamiltonian at $t_3, t_4$ is additive while the Hamiltonians of $MK$ and $B_2$ do not change. The third line follows from Lemma \ref{lemma:unitary-isothermal-work}. The final line follows from the assumption of erasure, i.e.,  $\rho_4^{MK} = \rho_0^{MK}$, so that $\Delta F^{MK}_{3 \to 4} = -\Delta F^{MK}_{0 \to 3}$, together with the fact that both the state and Hamiltonian of $MK$ do not change   between time steps $t_2$ and $t_3$, so that $-\Delta F^{MK}_{0 \to 3} = -\Delta F^{MK}_{0 \to 2} -\Delta F^{MK}_{2 \to 3} = -\Delta F^{MK}_{0 \to 2}$.  Given that $W_{\In, 0 \to 2}^{MK} = \Delta E^{MK}_{0 \to 2} = \Delta F^{MK}_{0 \to 2} + \beta^{-1} \Delta S_{0 \to 2}^{MK}$, we have that 
\begin{align}\label{eq:MK_in_2}
    W^{MK}_\In &= W_{\In, 0 \to 2}^{MK} + W_{\In, 3 \to 4}^{MK}  \nonumber \\      
    &= \Delta F^{MK}_{0 \to 2} + \beta^{-1} \Delta S_{0 \to 2}^{MK} - \Delta F^{MK}_{0 \to 2}  + \beta^{-1} [I(MK:B_2)_{\rho_4} + D(\rho_4 ^{B_2} \| \gamma^{B_2})] \nonumber \\
    & = \beta^{-1} [\Delta S_{0 \to 2}^{MK} + I(MK:B_2)_{\rho_4} + D(\rho_4 ^{B_2} \| \gamma^{B_2})].
\end{align}
Now note that in general, the following relationship holds:
    \begin{align}\label{eq:conditional-mutual-GO}
        I(A:M|K)_{\rho_2}&:=S(A|K)_{\rho_2}+S(M|K)_{\rho_2}-S(AM|K)_{\rho_2} \nonumber \\
        &=S(A|K)_{\rho_2} - S(A)_{\rho_0} + S(A)_{\rho_0} + S(MK)_{\rho_2} - S(MK)_{\rho_0} + S(MK)_{\rho_0} - S(AMK)_{\rho_2} \nonumber \\
         &=S(A|K)_{\rho_2} - S(A)_{\rho_0}  + S(MK)_{\rho_2} - S(MK)_{\rho_0} + S(AMK)_{\rho_0} - S(AMK)_{\rho_2} \nonumber \\
        &=-I_{\GO}+\Delta S^{MK}_{0\to2}-\Delta S^{AMK}_{0\to2}\;.
    \end{align}
The second line is obtained by adding and subtracting $S(A)_{\rho_0}$ and $S(MK)_{\rho_0}$, together with the definition $S(AM|K)_{\rho_2}:= S(AMK)_{\rho_2} - S(K)_{\rho_2}$ and $S(M|K)_{\rho_2}:= S(MK)_{\rho_2} - S(K)_{\rho_2}$. The third line is obtained by noting the fact that $\rho^{AMK}_0 = \rho^A_0 \otimes \rho^{MK}_0$ so that $S(A)_{\rho_0} + S(MK)_{\rho_0} = S(AMK)_{\rho_0}$. By combining Eq.~\eqref{eq:conditional-mutual-GO} and Eq.~\eqref{eq:MK_in_2}, we obtain the desired equality~\eq{eq:work_in_go}.

\subsection{Proof of Proposition \ref{prop:second-law-equiv}}\label{appendix:second-law-equiv}

By combining Eqs. \eqref{eq:extracable_holevo-go} and \eqref{eq:work_in_go}, we obtain
\begin{align}\label{eq:app-net-work}
W_{\textrm{tot}} &= W^A_\ext-W^{MK}_\In \nonumber \\
& = - \Delta F^{A}_{0 \to 4} -\beta^{-1} \bigg( \Delta S^{AMK}_{0 \to 2} + I(A:M|K)_{\rho_2} + I(A:K)_{\rho_3} + S_\mathrm{irr}^{B_1} + S_\mathrm{irr}^{B_2} \bigg).
\end{align}
Recall that the protocol is consistent with the overall second law of thermodynamics if and only if $W_{\textrm{tot}}\le-\Delta F^{AMK}_{0\to4}$. But now note that 
\begin{align} \label{eq:AMK-A}
    -\Delta F^{AMK}_{0\to 4}&=-\Delta F^A_{0\to 4} -\Delta F^{MK}_{0\to 4}-\beta^{-1}I(A:MK)_{\rho_4} \nonumber \\
    &=-\Delta F^A_{0\to 4}-\beta^{-1}I(A:MK)_{\rho_4} \nonumber \\
    & \leqslant -\Delta F^A_{0\to 4} \;,
\end{align}
where the first equality holds because of Lemma \ref{lemma:free_energy_mutual} and $I(A:MK)_{\rho_0}=0$, 
the second equality follows from the erasure condition $\rho^{MK}_4=\rho^{MK}_0$, and the  inequality follows from the non-negativity of the mutual information.   Then by \eq{eq:app-net-work}, the protocol is consistent with the overall second law if and only if 
\begin{align}\label{eq:overall-second-law-equiv-app}
   \Delta S^{AMK}_{0 \to 2} + I(A:M|K)_{\rho_2} + I(A:K)_{\rho_3} + S_\mathrm{irr}^{B_1} + S_\mathrm{irr}^{B_2} \geqslant I(A:MK)_{\rho_4}.
\end{align}
By rearranging the above, we obtain \eq{eq:overall-second-law-equiv}. Now recall that  the protocol is consistent with the  second law of ITh if and only if $W_{\textrm{tot}}\le-\Delta F^{A}_{0\to4}$. By the same arguments as before, only replacing $I(A:MK)_{\rho_4}$ in the right hand side of \eq{eq:overall-second-law-equiv-app} with $0$,  we obtain \eq{eq:second-law-ITh-equiv}. 

It is  clear that \eq{eq:overall-second-law-equiv} and \eq{eq:second-law-ITh-equiv} are equivalent if and only if erasure is perfect, that is, $\rho^{AMK}_4 = \rho^A_4 \otimes \rho^{MK}_0 = \rho^A_4 \otimes \rho^M_0 \otimes |0\>\<0|^K$, so that  $I(A:MK)_{\rho_4}=0$. But since in general $I(A:MK)_{\rho_4} \geqslant 0$, while \eq{eq:overall-second-law-equiv} always implies \eq{eq:second-law-ITh-equiv},  the converse implication does not always hold.  

\subsection{Proof of Proposition~\ref{prop:second-law-sufficient}}\label{appendix:second-law-sufficient}

To show that $\Delta S^{AMK}_{0\to2} \geqslant 0$ is sufficient for compatibility of the measurement process with the overall second law, we must show that the   right hand side of \eq{eq:overall-second-law-equiv} is never strictly positive. Given the non-negativity of the irreversible entropy production terms $S_\mathrm{irr}^{B_1}, S_\mathrm{irr}^{B_2}$,  it suffices to show that 
\begin{align*}
I(A:MK)_{\rho_4}   - I(A:M|K)_{\rho_2}  -I(A:K)_{\rho_3} \leqslant 0.    
\end{align*}
To this end, let us note that 
\begin{align}\label{eq:APP-2-3}
  I(A : M| K)_{\rho_2} &= \sum_{k \in \cK} p_k I(A:M)_{\rho_{k,2}} \nonumber \\
 &= \sum_{k \in \cK} p_k D( \rho^{AM}_{2,k} \| \rho^{A}_{2,k} \otimes \rho^{M}_{2,k} ) \nonumber \\
  & \geqslant  \sum_{k \in \cK} p_k D( \Lambda_k \otimes \idch^M( \rho^{AM}_{2,k}) \| \Lambda_k \otimes \idch^M(\rho^{A}_{2,k} \otimes \rho^{M}_{2,k}) ) \nonumber \\
  & = \sum_{k \in \cK} p_k D( \rho^{AM}_{3,k} \| \rho^{A}_{3,k} \otimes \rho^{M}_{3,k} ) \nonumber \\
  & =  \sum_{k \in \cK} p_k I(A:M)_{\rho_{k,3}} = I(A:M|K)_{\rho_3}.
\end{align}
Here, $\Lambda_k(\cdot) := \Tr_{B_1}[\cF_k(\gamma^{B_1} \otimes \cdot)]$ are the conditional channels  acting in $A$ during feedback,   the third line follows from the data processing inequality \cite{schumacher1996quantum}, and the fourth line follows from item (i) of Lemma \ref{lemma:local-energy-change}. Note that if feedback is pure unitary, so that $\Lambda_k(\cdot) = F_k^{A} (\cdot) F_k^{A\dagger}$, then the inequality above becomes an equality. 

Now notice that the following equality holds from the chain rule:
\begin{equation}\label{eq:APP-2-4}
    I(A:M|K)_{\rho_3}+I(A:K)_{\rho_3}=I(A:MK)_{\rho_3}\;.
\end{equation}
By \eq{eq:APP-2-3} and \eq{eq:APP-2-4}, it follows that 
\begin{align*}
I(A:MK)_{\rho_4}   - I(A:M|K)_{\rho_2}  -I(A:K)_{\rho_3} &\leqslant I(A:MK)_{\rho_4}   - I(A:M|K)_{\rho_3}  -I(A:K)_{\rho_3} \nonumber \\
& = I(A:MK)_{\rho_4} - I(A:MK)_{\rho_3} \nonumber \\
& = D(\rho_4^{AMK} \| \rho_4^{A} \otimes \rho_4^{MK} ) - D(\rho_3^{AMK} \| \rho_3^{A} \otimes \rho_3^{MK} ) \nonumber \\
& = D(\idch^A \otimes \Phi (\rho_3^{AMK}) \| \idch^A \otimes \Phi( \rho_3^{A} \otimes \rho_3^{MK}) ) \nonumber \\
& \qquad - D(\rho_3^{AMK} \| \rho_3^{A} \otimes \rho_3^{MK} ) \nonumber \\
& \leqslant 0.
\end{align*}
Here, $\Phi(\cdot) := \Tr_{B_2}[\cV(\cdot \otimes \gamma^{B_2})]$ is the erasure channel acting in $MK$, the fourth line follows from item (i) of Lemma \ref{lemma:local-energy-change}, and the final line follows from the data processing inequality.

Now, recall  from \propref{prop:second-law-equiv} that if a feedback control and erasure protocol is consistent with the overall second law, then it will necessarily also be consistent with the second law of ITh. Therefore, a measurement process satisfying $\Delta S^{AMK}_{0\to2} \geqslant 0$ is guaranteed to be compatible with the second law of ITh.

Finally,  we wish to show that if the instrument $\cM:= \{\cM_k : k \in \cK\}$ that is responsible for pointer objectification implements a bistochastic channel---a CP linear map that preserves both the trace and the unit---then $\Delta S^{AMK}_{0 \to 2} \geqslant 0$ will necessarily hold. Note that the channel implemented by $\cM$, i.e., $\cM_\cK(\cdot) := \sum_{k \in \cK} \cM_k(\cdot)$,  is bistochastic if $\cM_\cK(\one^M) = \one^M$. 

Recall that $\rho^{AMK}_2 = \sum_{k\in \cK} p_k \, \rho^{AM}_{2,k} \otimes |k\>\<k|^K$. Since the classical register $K$ is not entangled with $AM$, it follows that  $S(K|AM)_{\rho_2}\ge 0$. 
Thus, we have
\begin{align*}
\Delta S^{AMK}_{0\to 2} &= S(AMK)_{\rho_2}-S(AMK)_{\rho_0} \\
& =S(AM)_{\rho_2}+S(K|AM)_{\rho_2}-S(AM)_{\rho_0} \\
& \ge S(AM)_{\rho_2}-S(AM)_{\rho_0}    .
\end{align*}
 Given that unitary channels are bistochastic, then so long as the  channel  $\cM_\cK$ is also bistochastic, then so too is the composition $\Theta := (\idch^A \otimes \cM_\cK)\circ \cU$.  Now note that we may equivalently write the von Neumann entropy as $S(A)_\rho = - 
 D(\rho^A \| \one^A)$. As such, we have that 
 \begin{align*}
\Delta S^{AMK}_{0\to 2} & \geqslant S(AM)_{\rho_2} - S(AM)_{\rho_0}  \\
& = D( \rho^{AM}_0\| \one^{AM}) - D( \rho^{AM}_2 \| \one^{AM}) \\
& = D(\rho^A_0 \otimes \rho^M_0 \| \one^{AM}) - D(\Theta(\rho^A_0 \otimes \rho^M_0) \| \one^{AM}) \\
& = D(\rho^A_0 \otimes \rho^M_0 \| \one^{AM}) - D(\Theta(\rho^A_0 \otimes \rho^M_0) \| \Theta(\one^{AM})) \\
& \geqslant 0.
 \end{align*}
Here, in the fourth line we have used the bistochasticity of $\Theta$, and the  final line follows from the data processing inequality.

A paradigmatic example of an objectification process that is bistochastic is given by the L\"uders instrument. For any observable $\M$,  the operations of the corresponding $\M$-compatible  L\"uders instrument read $\cM^L_k(\cdot) := \sqrt{\M_k} (\cdot) \sqrt{\M_k}$. These are also known as ``square-root measurements''. It is clear that the channel implemented by a L\"uders instrument is bistochastic, since $\cM_\cK^L(\one^M) = \sum_{k \in \cK} \M_k = \one^M$. However,  every observable admits instruments that are not of the L\"uders type,  but which nonetheless implement a bistochastic channel---for example, the instrument with operations $\cM_k := \Phi \circ \cM^L_k$, where $\Phi$ is some arbitrary  bistochastic channel.

\subsection{Proof of Theorem \ref{theorem:second-law-necessary}}\label{appendix:second-law-necessary}
Let us note that the only term on the right hand side of  \eq{eq:second-law-ITh-equiv}  that is fixed by the measurement process alone is $-I(A:M|K)_{\rho_2}$.  Therefore, the right hand side of this equation, given a fixed measurement process but for all possible subsequent feedback and erasure processes,   is upper bounded as
\begin{align}\label{eq:RHS-second-law-ITh-upper-bound}
-I(A:M|K)_{\rho_2} -I(A:K)_{\rho_3} -S_\mathrm{irr}^{B_1} - S_\mathrm{irr}^{B_2} \leqslant -I(A:M|K)_{\rho_2}.    
\end{align}
This follows from the non-negativity of the mutual information and the entropy production terms. Imposing that the inequality in \eq{eq:second-law-ITh-equiv} must be satisfied even when the right hand side obtains the upper bound above, we thus arrive at \eq{eq:second-law-ITh-necessary-1}. Note that the upper bound of \eq{eq:RHS-second-law-ITh-upper-bound} is achievable in the limit where  feedback and erasure  are quasistatic, so that $S_\mathrm{irr}^{B_1} = S_\mathrm{irr}^{B_2}=0$,  and such that for all measurement outcomes  the feedback process transforms the target system to the same final state, i.e.,  $\rho^A_{3,k} = \rho^A_3$ for all $k$, so that $I(A:K)_{\rho_3} = 0$.  

To show that  \eq{eq:second-law-ITh-necessary-1} is equivalent to \eq{eq:second-law-ITh-necessary-2}, let us note that 
\begin{align}\label{eq:entropy-change-measurement-equiv}
\Delta S^{AMK}_{0 \to 2} &:= S(AMK)_{\rho_2} - S(AMK)_{\rho_0}  \nonumber  \\
& = S(AM|K)_{\rho_2} + S(K)_{\rho_2} - S(AMK)_{\rho_0} \nonumber \\
& = S(AM|K)_{\rho_2} + \mathscr{H}(\{p_k\}) - S(A)_{\rho_0} - S(M)_{\rho_0}.
\end{align}
In the second line we use the definition of the conditional entropy, whereas in the final line we use the fact that $\rho^K_{2} = \sum_k p_k |k\>\<k|^K$  and that $\rho_0^{AMK} = \rho^A_0 \otimes \rho^M_0 \otimes |0\>\<0|^K$. Moreover, let us note that by the definition of the conditional mutual information, it holds that 
\begin{align}\label{eq:cond-mutual-info-measurement-equiv}
-I(A:M|K)_{\rho_2} = -S(A|K)_{\rho_2} -S(M|K)_{\rho_2} + S(AM|K)_{\rho_2}  .  
\end{align}
By inserting Eqs. \eqref{eq:entropy-change-measurement-equiv} and \eqref{eq:cond-mutual-info-measurement-equiv} in \eq{eq:second-law-ITh-necessary-1} gives us \eq{eq:second-law-ITh-necessary-2}.

Finally, we shall  show that \eq{eq:second-law-ITh-necessary-1} is also necessary for the compatiblity of the measurement process with the overall second law. To this end, let us consider a feedback control and erasure protocol, and assume that the measurement process violates \eq{eq:second-law-ITh-necessary-1}, i.e., assume that $\Delta S^{AMK}_{0 \to 2} < -I(A:M|K)_{\rho_2}$, but such that the protocol is consistent with the overall second law, i.e., \eq{eq:overall-second-law-equiv}.  This gives us the inequality
\begin{align*}
 \Delta S^{AMK}_{0\to2} & >  \Delta S^{AMK}_{0\to2} +  I(A:MK)_{\rho_4}   -I(A:K)_{\rho_3} -S_\mathrm{irr}^{B_1} - S_\mathrm{irr}^{B_2}  .
\end{align*}
Assume also that the feedback and erasure processes are ideal and quasistatic, so that  $I(A:K)_{\rho_3} = S_\mathrm{irr}^{B_1} = S_\mathrm{irr}^{B_2} = 0$. In such a case the above inequality becomes
\begin{align*}
 \Delta S^{AMK}_{0\to2} & >  \Delta S^{AMK}_{0\to2} +  I(A:MK)_{\rho_4}.
\end{align*}
But by the non-negativity of the mutual information, this inequality  cannot be satisfied. As such, if a measurement process violates \eq{eq:second-law-ITh-necessary-1}, then it will necessarily violate \eq{eq:overall-second-law-equiv} for some feedback and erasure process. It follows that \eq{eq:second-law-ITh-necessary-1} is necessary for compatibility of the measurement process with the overall second law. 

\subsection{A measurement process that is incompatible with the second laws} \label{appendix:bistochastic}

Recall from  \propref{prop:second-law-sufficient} that a necessary condition for the incompatibility of the measurement process  with the second laws is that pointer objectification must not be bistochastic.  This is always possible; for example, a measure and prepare instrument $\cM_k(\cdot) = \Tr[\cM_k( \cdot)] |\psi\>\<\psi|^M$, where $|\psi\>^M$ is a fixed,  arbitrary pure state of $M$. It is trivial that $\cM_\cK$ is not bistochastic, since $\cM_\cK(\one^M) = \Tr[\one^M] |\psi\>\<\psi|^M \ne \one^M$.  We shall now use a measurement process utilizing just such a pointer objectification, demonstrating that it is incompatible with the second laws. 

Let $(\cH^M, \rho^M_0, \cU,\cM)$ be a measurement process for a L\"uders instrument $\cA^L_k(\cdot) := \A_k (\cdot) \A_k$, compatible with a projection valued measure $\A$,  acting in the target system. Here, we choose $\cM$ to be compatible with a projection valued measure $\M$, and we choose $\rho^M_0$ to be a mixed state, albeit of sufficiently low rank so that our model is in accordance with \propref{prop:rank-requirement-efficient-instrument}. Recall that any instrument $\cM$ that is compatible with the same POVM will realise the same instrument acting in the target system. Therefore, let us choose this instrument to be nuclear, with the operations $\cM_k(\cdot) = \Tr[\M_k \cdot] |\psi\>\<\psi|^M$, where $|\psi\>^M$ is a fixed,  arbitrary pure state of $M$.   Now let us choose one particular outcome $k=h$, and choose the input state of the target system  so that it has support only in the eigenvalue-1 eigenspace of the effect $\A_h$, i.e., $\A_k \rho^A_0 = \delta_{k,h}\rho^A_0$. In such a case, it will hold that $p_k = \delta_{k,h}$, and so $\mathscr{H}(\{p_k\}):= -\sum_{k \in \cK} p_k \ln p_k  = 0$. Moreover, we have that $\rho^A_{2,k} = \delta_{k,h} \rho^A_0$, so that $I_{\GO} = 0$. But, given the choice of instrument $\cM$ acting in the memory, it holds that $\rho^M_{2,k} = \delta_{k,h} |\psi\>\<\psi|^M$, so that $J_{\GO} = S(M)_{\rho_0} >0$. Our protocol therefore gives the inequality
\begin{align*}
\mathscr{H}(\{p_k\}) <    I_{\GO} +  J_{\GO},
\end{align*}
which contradicts \eq{eq:second-law-ITh-necessary-2} and so, by  \thmref{theorem:second-law-necessary}, violates the second law of ITh and the overall second law for some feedback and erasure processes. Indeed, note also that in this model, we have $\rho^{AMK}_2 = \rho^A_0 \otimes |\psi\>\<\psi|^M \otimes |h\>\<h|^K$, where the lack of correlations between $A$ and $M$ follows from the fact that $\cM$ is a nuclear instrument. In such a case, it holds that $\Delta S^{AMK}_{0\to 2} = -S(M)_{\rho_0} <0$, whereas $-I(A:M|K)_{\rho_2} = -I(A:M)_{\rho_2} = 0$. It follows that  
\begin{align*}
\Delta S^{AMK}_{0\to 2} <    -I(A:M|K)_{\rho_2},
\end{align*}
which contradicts \eq{eq:second-law-ITh-necessary-1}.

\subsection{Efficient instruments}\label{appendix:efficient-instrument-rank}

\begin{proposition}\label{prop:rank-requirement-efficient-instrument}
 Let $(\cH^M, \rho^M_0, \cU, \cM)$ be a measurement scheme for an  instrument $\cA$ compatible with an observable $\A := \{\A_k : k = 0, \dots, N\}$ acting in $A$, where $N$ is the number of distinct measurement outcomes, and where $\A_0 = \zero^A$ is a null effect. Assume that $\cM$ is compatible with a projection valued measure $\M := \{\M_k : k=0, \dots, N\}$ acting in $M$,  and denote $\cH^{ M_k} :=  \supp(\M_k)$.  Assume that the effects of $\A$, excluding the null effect $\A_0$, are linearly independent. Then $\cA$ is efficient only if 
\begin{align*}
 \rank{\rho^M_0} \leqslant \frac{\sum_{k=1}^N\dim(\cH^{ M_k} )}{N} \leqslant \frac{\dim(\cH^M)}{N},   
\end{align*}
with the second inequality becoming an equality  if and only if $\M_0 = \zero^M$.
\unskip\nobreak\hfill $\square$
\end{proposition}
\begin{proof}
Note that Assumption (A-6) assumes that the outcome associated with projecting $M$ onto the subspace $\cH^{M_0}$ is (statistically) never observed, i.e., it is observed with probability zero.  For this reason, in what follows, we need to introduce the effect $\M_0$ of the pointer observable, associated with a null effect $\A_0=\zero^A$ for the system observable, which makes the presentation a little cumbersome.

To prove the claim, we first note that an efficient instrument compatible with an observable with linearly independent effects is \emph{extremal} \cite{Pellonpaa2013}; given the instruments $\cA, \cA', \cA ''$, all with the same value space $\cK$, $\cA$ is extremal if for any  $\lambda \in (0, 1)$, we may write $\cA_k(\cdot) = \lambda \cA_k'(\cdot) + (1 - \lambda) \cA_k''(\cdot)$ only if $\cA = \cA'= \cA''$. That is, an instrument $\cA$ is extremal if it cannot be written as a convex combination of distinct instruments.  As such, we shall first obtain necessary conditions on the rank of $\rho^M_0$ that must be satisfied for  the measurement scheme to implement a  general extremal instrument $\cA$.

Let us write $\rho^M_0 = \sum_{i=1}^r q_i |\phi_i\>\<\phi_i|$, where $|\phi_i\>$ are mutually orthogonal unit vectors, $\{q_i\}$ is a probability distribution, and $r = \rank{\rho^M_0}$. By linearity, for each $i$ it holds that $(\cH^M, |\phi_i\>, \cU, \cM)$ is a measurement scheme for an instrument $\cA^{(i)}$, such that $\sum_i q_i \cA_k^{(i)}(\cdot) = \cA_k(\cdot)$ for all $k$. Note that since outcome $k=0$ of the pointer observable is associated with the null effect $\A_0 = \zero^A$, then it holds that $\cA_0 (\cdot) = \cA_0^{(i)}(\cdot) = \zero^A$. Denoting the (projection) effects of the pointer observable $\M$ as $\M_k = \sum_\mu |\psi_{k, \mu}\>\<\psi_{k, \mu}|$, where $\{|\psi_{k, \mu}\>\}$ is an orthonormal basis that spans $\cH^{M}$, then for each $i$ and $k$, by \eq{eq:system-instrument} we may write 
\begin{align*}
\cA_k^{(i)} (\cdot) = \Tr_M[(\one^A \otimes \M_k) U(\cdot \otimes |\phi_i\>\<\phi_i|)U^\dagger] = \sum_\mu  L_{k, \mu}^{(i)} (\cdot) 
 L_{k, \mu}^{(i)^\dagger},   
\end{align*}
where the Kraus operators read
\begin{align*}
L_{k, \mu}^{(i)} = V_{\psi_{k,\mu}}^\dagger U V_{\phi_i}.    
\end{align*}
 Here,  $V_{\varphi} : \cH^A \to \cH^A \otimes \cH^M, \ket{\xi} \mapsto \ket{\xi} \otimes \ket{\varphi}$ are linear isometries defined by the unit vector $|\varphi\> \in \cH^M$, which satisfy 
\begin{align*}
& V_{\varphi}^\dagger \one^{AM} V_{\varphi'} = \<\varphi|\varphi'\> \one^A, &V_{\varphi} \one^A V_{\varphi'}^\dagger = \one^A \otimes |\varphi\>\<\varphi'| .   
\end{align*}
Noting that $\sum_{k,\mu} |\psi_{k,\mu}\>\<\psi_{k,\mu}| = \one^M$, it follows that for every $i\ne j$, it holds that 
\begin{align}\label{eq:Kraus-null}
 \sum_{k,\mu}   L_{k, \mu}^{(i)\dagger } L_{k, \mu}^{(j)} &=  \sum_{k,\mu}  V_{\phi_i}^\dagger U^\dagger V_{\psi_{k,\mu}} \one^A V_{\psi_{k,\mu}}^\dagger  U V_{\phi_j}  =  V_{\phi_i}^\dagger \one^{AM} V_{\phi_j} =\zero.
\end{align}
 Let $\{L_{k,\nu} \, | \, \nu = 1, \dots , R_k\}$ be a minimal Kraus representation for the operation $\cA_k$, i.e.,  where $L_{k,\nu}$ are linearly independent and  $R_k$ is the Kraus-rank of $\cA_k$. Note that since $\A_0=\zero^A$, then $L_{0,\nu}=\zero^A$. Now assume that $\cA$ is an extremal instrument. This implies that $\cA_k = \cA_k^{(i)}$ for all $i$ and $k$. As shown in \cite{Choi1975}, for each $i$  there exists an isometry $[u_{\mu,\nu}^{(i)} \in \mathds{C}]$ such that
\begin{align}\label{eq:isometry-Kraus}
& L_{k,\mu}^{(i)} = \sum_\nu u_{\mu, \nu}^{(i)} L_{k,\nu},   & \sum_{\mu} u_{\mu, \nu}^{(i)*} u_{\mu, \nu'}^{(i)} = \delta_{\nu, \nu'}.
\end{align}
By  \eq{eq:Kraus-null}, \eq{eq:isometry-Kraus}, and orthonormality of $\{| \psi_{k, \mu}\>\}$, we may thus write for every $i\ne j$ the following:
\begin{align}\label{eq:extremality-Kraus-equality}
  \zero &=  \sum_{k,\mu, \mu'}   L_{k, \mu}^{(i)\dagger } L_{k, \mu '}^{(j)} \<\psi_{k,\mu}| \psi_{k, \mu'}\> \nonumber \\
  & = \sum_{k,\mu, \mu'}  \bigg(\sum_\nu u_{\mu, \nu}^{(i)*} L_{k,\nu}^\dagger \bigg) \bigg( \sum_{\nu'} u_{\mu', \nu'}^{(j)} L_{k,\nu'}\bigg)  \<\psi_{k,\mu}| \psi_{k, \mu'}\> \nonumber \\
  & = \sum_{k, \nu, \nu'}  L_{k,\nu}^\dagger L_{k,\nu'}  \<\psi_{k,\nu}^{(i)}| \psi_{k, \nu'}^{(j)}\>,
\end{align}
where
\begin{align}\label{eq:isometry-pointer-basis}
|\psi_{k,\nu}^{(i)}\> := \sum_{\mu}  u_{\mu, \nu}^{(i)} | \psi_{k, \mu}\> \in \supp(\M_k) \equiv \cH^{M_k}.   
\end{align}
As shown in \cite{DAriano2011}, $\cA$ is an extremal instrument if and only if the set 
\begin{align*}
\{ L_{k,\nu}^\dagger L_{k,\nu'} \, | \, k = 1, \dots, N ; \nu,\nu'= 1, \dots, R_k \}    
\end{align*} 
is linearly independent. As such, the equality condition in \eq{eq:extremality-Kraus-equality} holds only if  $\<\psi_{k,\nu}^{(i)}| \psi_{k, \nu'}^{(j)}\> =0$ for all $k >0$,  $\nu, \nu'$, and $i\ne j$. Now, by \eq{eq:isometry-Kraus} and \eq{eq:isometry-pointer-basis}, together with the fact that $\<\psi_{k,\mu}|\psi_{k', \mu'}\> = \delta_{k,k'} \delta_{\mu, \mu'}$, it is easily verified that $\<\psi_{k,\nu}^{(i)}|\psi_{k',\nu'}^{(i)}\> =\delta_{k,k'} \delta_{\nu, \nu'}$ for every $i$. Indeed, since for every $i$,  $|\psi_{k,\nu}^{(i)}\> \in \cH^{M_k}$,  then it also holds that $\<\psi_{k,\nu}^{(i)}|\psi_{k',\nu'}^{(j)}\>= 0$ whenever $k\ne k'$.  It follows that  
\begin{align*}
   \{|\psi_{k,\nu}^{(i)}\> \in \bigoplus_{k=1}^N \cH^{M_k} \, | \,  k =1, \dots, N ; \nu = 1, \dots , R_k ; i = 1, \dots , \rank{\rho^M_0}\} 
\end{align*}
must be a set of mutually orthogonal vectors.  The cardinality of the above set is easily computed to be $\rank{\rho^M_0} \sum_{k=1}^N R_k$. But since $\bigoplus_{k=1}^N \cH^{M_k}$ can only contain at most $\dim(\bigoplus_{k=1}^N \cH^{M_k}) = \sum_{k=1}^N \dim(\cH^{M_k})$ mutually orthogonal vectors, then $\cA$ is extremal only if 
\begin{align*}
 \rank{\rho^M_0} \leqslant \frac{\sum_{k=1}^N \dim(\cH^{M_k})}{\sum_{k=1}^N R_k}.
\end{align*}
Now assume that $\cA$ is an efficient instrument. It holds that $R_k = 1$ for each $k$, and $\cA$ is an extremal instrument if and only if  $\{ L_{k}^\dagger L_{k} = \A_k \, | \, k = 1, \dots , N \}$, i.e., the non-trivial effects of the measured observable $\A$ in $A$, are linearly independent. This completes the proof.  
\end{proof}

\section*{Data availability}

No datasets were generated or analysed during the current study.

\bibliography{myref}

\section*{Acknowledgments}
The authors would like to thank Arshag Danageozian, Marco Genoni, Masahito Hayashi, Kenta Koshihara, Yosuke Mitsuhashi, Nelly H.Y. Ng, Yoshifumi Nakata, Takahiro Sagawa, Valerio Scarani, Jeongrak Son, and Philipp Strasberg for their helpful comments and fruitful discussions.
S.~M. acknowledges the ``Nagoya University Interdisciplinary Frontier Fellowship'' supported by Nagoya University and JST, the establishment of university fellowships towards the creation of science technology innovation, Grant Number JPMJFS2120 and ``THERS Make New Standards Program for the Next Generation Researchers'' supported by JST SPRING, Grant Number
JPMJSP2125.
M. H. M. acknowledges support from the European Union under project ShoQC within ERA-NET Cofund in Quantum Technologies (QuantERA) program, from the Slovak Academy of Sciences under IMPULZ project No. IM-2023-79 (OPQUT), as well as from projects VEGA 2/0183/21 (DESCOM) and  APVV-22-0570 (DeQHOST).
K.~K. acknowledges support from JSPS Grant-in-Aid for Early-Career Scientists, No. 22K13972; from MEXT-JSPS Grant-in-Aid for Transformative Research Areas (A) ``Extreme Universe'', No. 22H05254.
F.~B. acknowledges support from MEXT Quantum Leap Flagship Program (MEXT QLEAP) Grant No.~JPMXS0120319794, from MEXT-JSPS Grant-in-Aid for Transformative Research Areas (A) ``Extreme Universe'' No.~21H05183, and from JSPS KAKENHI, Grants No.~20K03746 and No.~23K03230.

\section*{Author contributions}

S.~M. and M.~H.~M. contributed equally to the conception of the ideas and the execution of the work. K.~S. contributed to the initial formulation of the ideas. K.~K. contributed to clarifying the logic of the arguments presented. F.~B. contributed to the conception of the ideas and the execution of the work. All authors contributed to the interpretation and discussion of the results.

\section*{Competing interests}
The authors declare no competing interests.

\end{document}